\documentclass[final,onefignum,onetabnum]{siamart190516}

\usepackage{ifthen}
\newboolean{arxiv}
\setboolean{arxiv}{false}

\ifthenelse{\boolean{arxiv}}{
 \usepackage[margin=2cm]{geometry}
 \usepackage{amsmath, amsthm, amssymb, url, color, enumitem}
}{
 \usepackage{amsmath}
 \usepackage{amsfonts}
 \usepackage{enumitem}
}

\usepackage{tikz}

\usepackage{subfig}

\usepackage[draft]{changes}
\definechangesauthor[name=TW, color=blue]{TW}
\definechangesauthor[name=MG, color=red]{MG}

\newcommand{\BF}[1]{{\bf\boldmath{#1}\unboldmath}}

\newcommand{\boundary}{\partial}
\newcommand{\subcube}{\Delta}

\newcommand{\preclassified}[1]{\tilde{#1}}
\newcommand{\classified}[1]{{#1}^*}

\newcommand{\cv}{\mathrm{cv}}
\newcommand{\cs}{\mathrm{cs}}
\newcommand{\dv}{\mathrm{dv}}
\newcommand{\ds}{\mathrm{ds}}
\newcommand{\dr}{\mathrm{dr}}
\newcommand{\ax}{\mathrm{a}}
\newcommand{\bx}{\mathrm{b}}
\newcommand{\as}{\mathrm{as}}

\newcommand{\class}{\mathrm{class}}
\newcommand{\grid}{\mathbb{G}}
\newcommand{\tree}{\mathbb{T}}
\newcommand{\hypercube}{\mathbb{H}}
\newcommand{\boxes}{\mathbb{B}}
\newcommand{\refines}{\preceq}

\newcommand{\shapeclassregular}{\mathrm{H}}
\newcommand{\shapeclasscells}{\Theta}

\newcommand{\classregular}{\mathrm{K}}
\newcommand{\classcells}{\Lambda}

\newcommand{\Cantor}{\mathrm{Cantor}}

\newcommand{\Z}{\mathbb{Z}}
\newcommand{\shape}{\mathrm{shape}}
\newcommand{\depth}{\mathrm{depth}}
\newcommand{\content}[1]{ \langle #1 \rangle }
\newcommand{\parent}[1]{\breve{#1}}

\newcommand{\diameter}{\delta}

\newcommand{\Shapes}{\mathrm{Shapes}}
\newcommand{\Partitions}{\mathrm{Partitions}}

\newcommand{\Define}[1]{\textit{#1}}

\numberwithin{equation}{section}

\ifthenelse{\boolean{arxiv}}{
 \theoremstyle{plain}
 \newtheorem{theorem}	[equation]	{Theorem}
 \newtheorem{corollary}	[equation]	{Corollary}
 
 \newtheorem{lemma}		[equation]	{Lemma}
 
}{
}

\newtheorem{example}	[equation]	{Example}

\newtheorem{claim}		[equation]	{Claim}

\newtheorem{rationale}    [equation]	{Rationale}
\newtheorem{application}  [equation]	{Application digression}

\newcommand{\TheTitle}{The maximum discrete surface-to-volume ratio of space-filling curve partitions}

\ifthenelse{\boolean{arxiv}}{
}
{
 \title{
   \TheTitle
   \thanks{Submitted to the editors DATE.
   \funding{Sponsed by EPSRC under the Excalibur Phase I call through grant number EP/V00154X/1 (ExaClaw) and EP/V001523/1 (SPH).}}
 }
  \author{
    Maximilien Gadouleau
    \thanks{
     Department of Computer Science, 
     Durham University
     (\email{m.r.gadouleau@durham.ac.uk}).
    }
    \and 
    Tobias Weinzierl
    \thanks{
     Department of Computer Science, 
     Durham University
     (\email{tobias.weinzierl@durham.ac.uk}).
    }
  }

  \headers{Maximum discrete surface-to-volume ratio of SFC partitions}{Maximilien Gadouleau and Tobias Weinzierl}

\ifpdf
\hypersetup{ pdftitle={Maximum discrete surface-to-volume ratio of SFC partitions} }
\fi
}

\begin{document}

\ifthenelse{\boolean{arxiv}}{
\title{\TheTitle}
\author{Maximilien Gadouleau and Tobias Weinzierl}
}
{}

\maketitle

\ifthenelse{\boolean{arxiv}}{
 \begin{abstract}
  Space-filling curves (SFCs) are used in high performance computing to
distribute a computational domain
or its mesh, respectively, amongst different compute units, i.e.~cores or
nodes or accelerators. 
The part of the domain allocated to each compute unit is called a partition.
Besides the balancing of the work, the communication cost to exchange
data between units determines the quality of a chosen partition.
This cost can be approximated by the surface-to-volume ratio of partitions: 
the volume represents the amount of local work, while the
surface represents the amount of data to be transmitted. 
Empirical evidence suggests that space-filling curves yield advantageous
 surface-to-volume ratios.
Formal proofs are available only for regular grids.
 We investigate the surface-to-volume ratio of space-filling curve partitions for
adaptive grids and derive the maximum surface-to-volume ratio as a
function of the number of cells in the partition. In order to prove our main theorem, we construct a new
framework for the study of adaptive grids, notably introducing the concepts of a
shape and of classified partitions.
The new methodological framework yields insight about the
 SFC-induced partition character even if the grids refine rather aggressively
 in localised areas:
 it quantifies the obtained surface-to-volume ratio.
 This framework thus has the potential to guide the design of
 better load balancing algorithms on the long term.

 \end{abstract}
}{
 \begin{abstract}
  Space-filling curves (SFCs) are used in high performance computing to
distribute a computational domain
or its mesh, respectively, amongst different compute units, i.e.~cores or
nodes or accelerators. 
The part of the domain allocated to each compute unit is called a partition.
Besides the balancing of the work, the communication cost to exchange
data between units determines the quality of a chosen partition.
This cost can be approximated by the surface-to-volume ratio of partitions: 
the volume represents the amount of local work, while the
surface represents the amount of data to be transmitted. 
Empirical evidence suggests that space-filling curves yield advantageous
 surface-to-volume ratios.
Formal proofs are available only for regular grids.
 We investigate the surface-to-volume ratio of space-filling curve partitions for
adaptive grids and derive the maximum surface-to-volume ratio as a
function of the number of cells in the partition. In order to prove our main theorem, we construct a new
framework for the study of adaptive grids, notably introducing the concepts of a
shape and of classified partitions.
The new methodological framework yields insight about the
 SFC-induced partition character even if the grids refine rather aggressively
 in localised areas:
 it quantifies the obtained surface-to-volume ratio.
 This framework thus has the potential to guide the design of
 better load balancing algorithms on the long term.

 \end{abstract}

}

\section{Introduction}

%
%
Space-filling curves \cite{Bader:2013:SFCs,Sagan:94:SFCs} (SFCs) are an elegant
paradigm
to linearise (high-dimensional) data in scientific computing,
database storage, imaging, and so forth.
 A linearisation makes it straightforward to cut data into chunks of equal size, i.e.~to realise a data (domain) decomposition.
Indeed, many simulation codes and load balancing libraries use SFCs as partitioning algorithm or as partitioning heuristic,
as
SFCs are conceptually simple; for an overview see \cite[Chapter
10]{Bader:2013:SFCs}. The most popular curves are Lebesgue
and Hilbert which work on topologically
hypercubic domains.
 The former has been ``rediscovered'' as Morton order
 \cite{Morton:66:SFC,Schrack:2015:GenerateSFCs} and is sometimes also called z-order.

\begin{figure}[htb]
 \centering
 \includegraphics[width=0.3\textwidth]{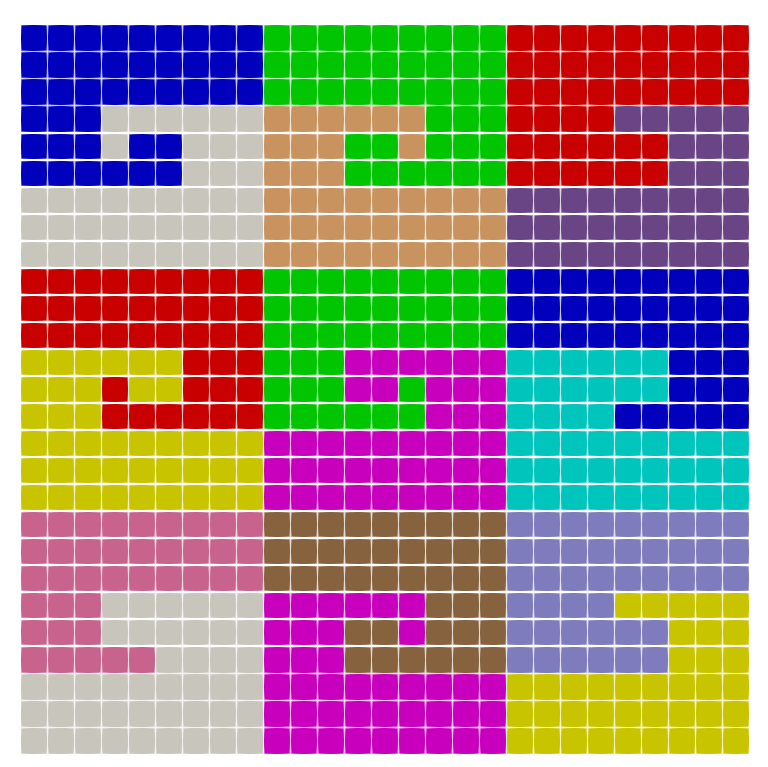}
 \hspace{0.2cm}
\begin{tikzpicture}[scale=0.48]
	\draw[step = 4cm] (0,0) grid (8,8);

	\draw[step = 2cm] (0,0) grid (4,4);

	\draw[step = 1cm] (0,0) grid (2,2);

	\draw[step = 0.5cm] (0,0) grid (1,1);
\end{tikzpicture}
  \caption{
    Left: Typical SFC-based domain decomposition. Here, the
    Peano SFC is used and we discretise a regular grid over the unit square.
    Right: The corner adaptive grid ($M=4$, $d=k=2$), where both the discrete
    surface and volume grow linearly with the depth $M$.
    \label{fig:introduction:corner}
    \vspace{-0.6cm}
  } 
\end{figure}
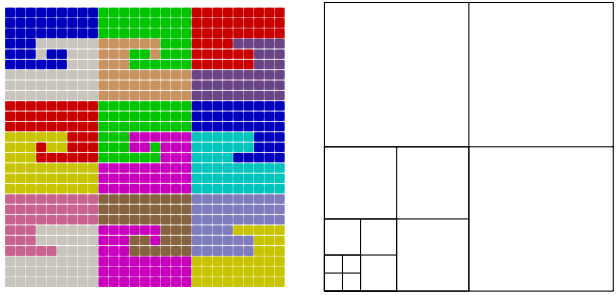

%
%
These curves take a $d$-dimensional hypercube 
which overlaps the data of interest---this can be a point cloud or a mesh simulating fluid flow, e.g.---and subdivide it into $k^d$
cells of equal size, with $k=2$ for Hilbert and Lebesgue and $k=3$ for the Peano SFC.
Even these two choices allow us to introduce a
multitude of different curves
\cite{Cardona:2016:SelfSimilarity,Haverkort:2010:LocalityFor2DSFCs,Schrack:2015:GenerateSFCs}.
They then continue subdividing the cells recursively.
The decision whether to stop the recursion or not is made per cell per
 refinement step.
This process yields a cascade of refinement levels.
Each refinement step orders the cells (little hypercubes) along a suitably
rotated or mirrored motif. We obtain a one-dimensional order per resolution level.
Historically, the continuous (space-filling) curve resulting from the limit of this construction pattern has been 
of interest in topology.
In the present context of scientific computing, we are interested in orderings over hypercubes of finite size---usually 
called mesh or grid---where not all cubes have the same size.
 We focus on discrete SFCs over adaptive mesh refinement (AMR).

It is straightforward to cut the (one-dimensional) sequence 
of cells along the SFC into chunks of roughly the same cell count.
This yields a \Define{partition} of the $d$-dimensional domain
(Figure~\ref{fig:introduction:corner}).
The term partition does not necessarily imply that all subdomains host the same
cell count: cells might induce different computational load or the load
balancing of choice might have introduced geometrically imbalanced subdomains.

%
%
%
The motivation behind domain decomposition in parallel scientific computing is to
balance the workload between compute units, plus, at the same time, to keep
the number of data exchanges between partitions small.
For algorithms that work with meshes, the latter criterion typically translates
into a small surface-to-volume ratio.
The larger the volume of a partition, i.e.~the higher the cell count, the more
work is to be done locally on a parallel computer.
The larger the surface (face count) however, the more data has to be exchanged
with other partitions.
Another term for the latter characterisation is edge-cuts: We assume that cells
that are neighbours have to communicate witch each other.
This imposes a communication graph between the cells of a mesh.
When we cut the mesh into partitions, we cut through the edges of this
communication graph.
 The arising problem to cut a mesh into equally sized partitions with minimal edge cuts, i.e.~minimal surface,
 is NP-hard \cite{Zumbusch:2001:QualityOfSFCs}.
SFCs are popular as they yield good surface-to-volume ratios ``for free''. 
They provide a good partitioning heuristic (cmp.~\cite{Bader:2013:SFCs,Gotsman:1996:MetricsOnSFCs,Haverkort:2010:LocalityFor2DSFCs,Harlacher:2012:DynamicSFCBalancing,Hungershoefer:02:SFCQuality,Samfass:2020:LoadBalancing,Sasidharan:2015:SFCsForPartitioning,Touheed:2000:Comparison,Zumbusch:2001:QualityOfSFCs}):
As the ordering of the cells is given, there is no freedom how to cut the
linearisation.
A user can only decide where to cut.
We face a classic chains-on-chains partitioning (CCP) challenge \cite{Harlacher:2012:DynamicSFCBalancing,Meister:2016:Samoa,Pinar:2008:ChainsOnChains}.
The surface of the partition deduces automatically.
 Further to the simplicity and determinism of SFC cuts, SFCs are popular as cutting SFCs requires limited global 
 knowledge and, hence, information exchange
 \cite{Clevenger:2020:MultigridFEM,Harlacher:2012:DynamicSFCBalancing}. 

%
%
Yet, this ``very good'' is empirical knowledge which is supported by proofs for
regular meshes only, i.e.~for meshes where all (hyper-)cubes have exactly the same size
\cite{Bungartz:2006:ParallelAdaptivePDESolver,Hungershoefer:02:SFCQuality,Zumbusch:2001:QualityOfSFCs}. 
 ``Very good'' has a continuous equivalent
by studying the $k$-boxes and relating them to their perimeter
 \cite{Haverkort:2010:LocalityFor2DSFCs}.
Let the \Define{discrete volume} $\dv(P)$ of a partition $P$ be the number of
cells within this partition that are not subdivided further. 
The \Define{discrete surface} $\ds(P)$ of the partition is the number of boundary
faces.
In a regular grid, where all cells have the same size, the discrete surface
$\ds(P)$ and the discrete volume $\dv(P)$ are proportional to their respective
continuous counterparts $\cv(P)$ and $\cs(P)$.  There exists a constant $C$ which only depends on the SFC such that

\begin{equation}
  \ds(P) \le C \dv(P)^{1-1/d}
  \label{equation:introduction:quasi-optimal}
\end{equation}
for any partition $P$. 
The inequality assumes reasonably detailed meshes, i.e.~reasonably big $\dv(P)$
and, hence, $\ds(P)$.
The
resulting relation is also called \Define{quasi-optimal}, as it equals,
besides a constant, the surface-to-volume ratio of a hypersphere
\cite{Bungartz:2006:ParallelAdaptivePDESolver,Gotsman:1996:MetricsOnSFCs,Haverkort:2010:LocalityFor2DSFCs,Zumbusch:2001:QualityOfSFCs}.
No geometric object can have a more advantageous shape with $ \cs(P) \ll C \cv(P)^{1-1/d}$.

The direct relation between discrete and continuous measures ceases
to exist for adaptive grids:
 A trivial worst-case example refines adaptively towards a corner of a
 partition (Figure \ref{fig:introduction:corner}), and consequently yields a
 linear relation between discrete surface and volume.
 The upper bound (\ref{equation:introduction:quasi-optimal}) is too strong.
 The other way round, we start from a given partition and refine exclusively
 cells that are not adjacent to the partition boundary.
 In this case, the quasi-optimal upper bound is very pessimistic.


\begin{figure}[htb]
 \begin{center}
  \includegraphics[width=0.4\textwidth]{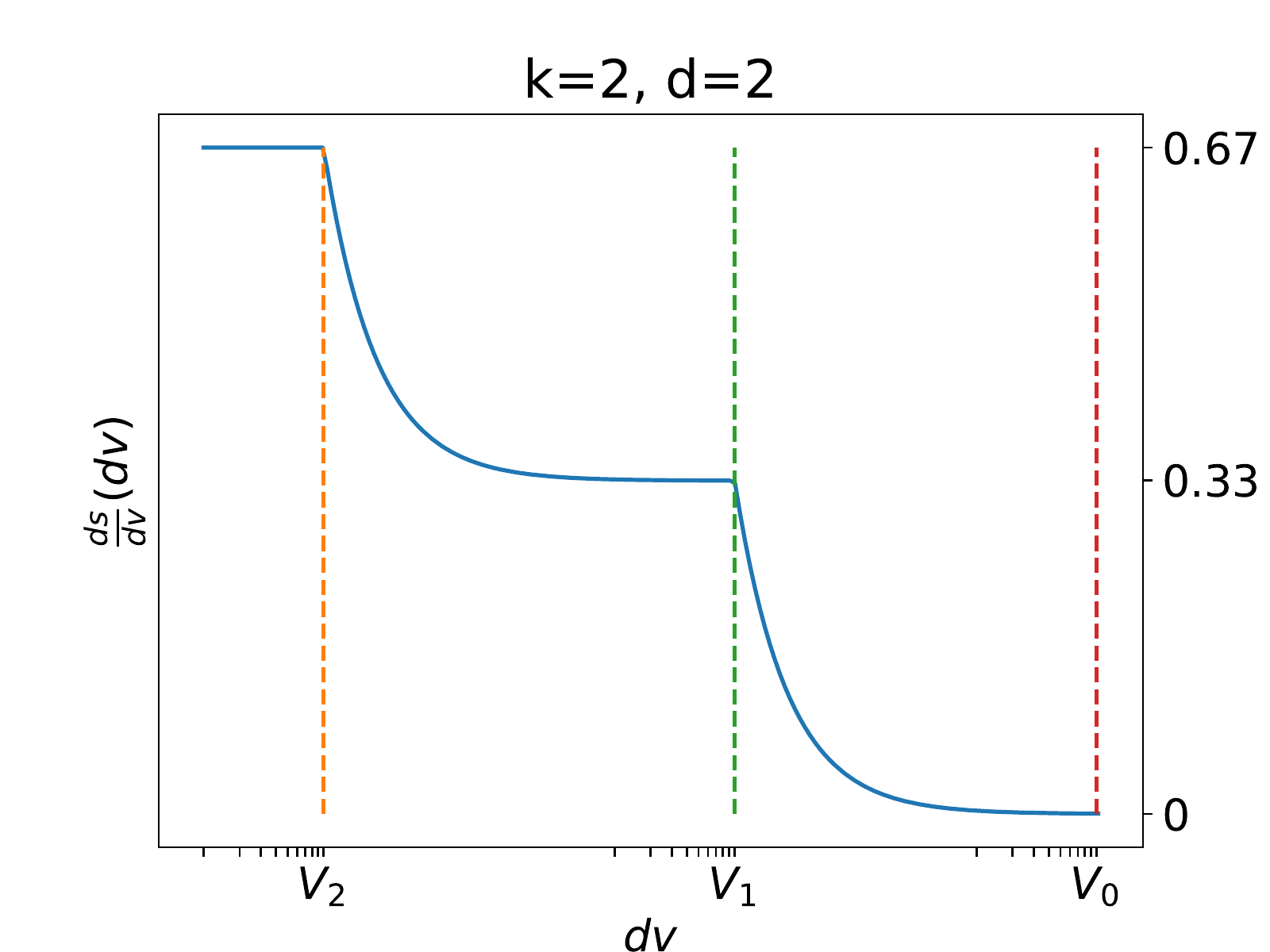}
  \includegraphics[width=0.4\textwidth]{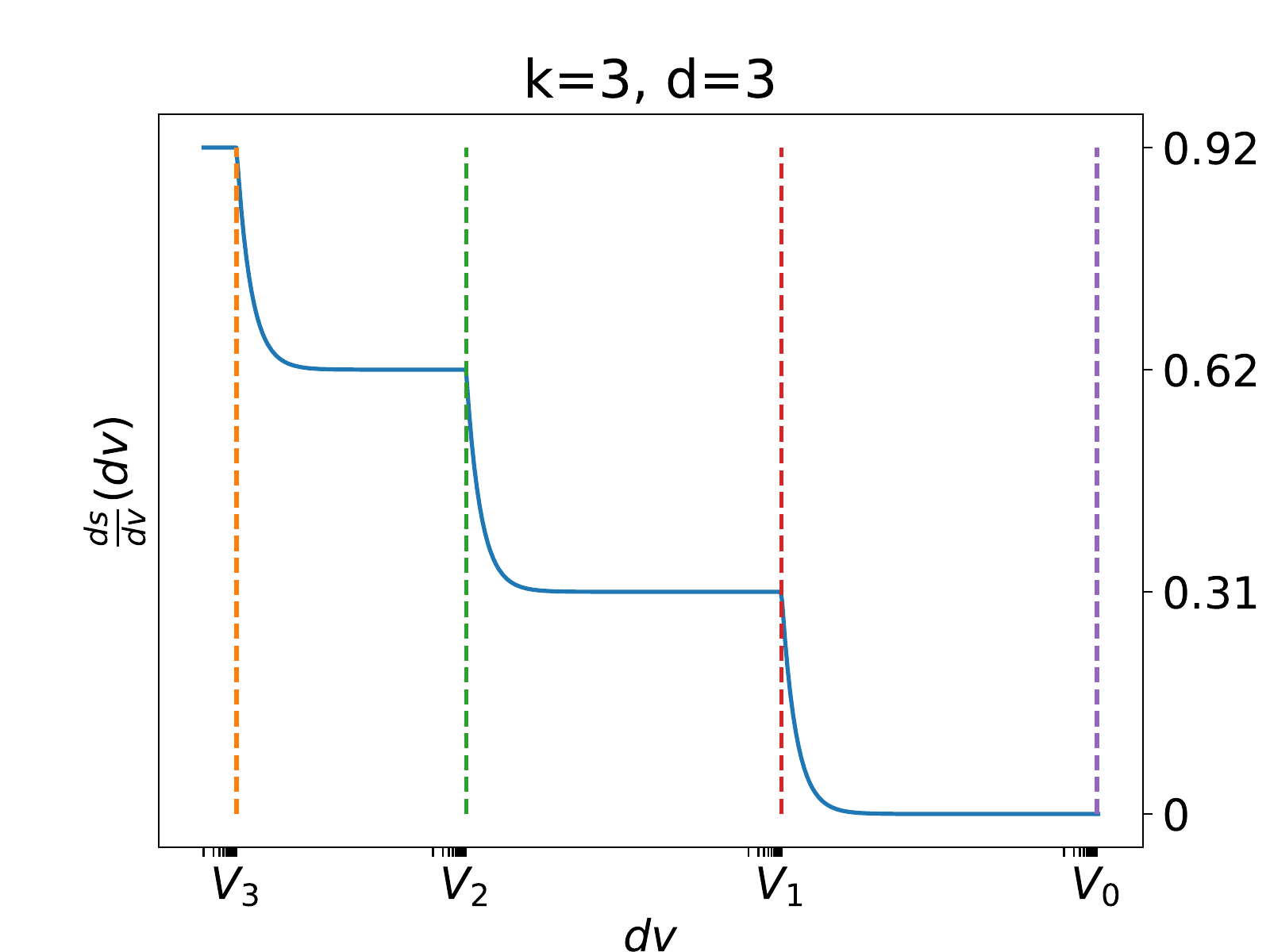}
 \end{center}
 \caption{
  Qualitative sketch of Theorem 
  \ref{th:surface_to_volume_asymptotic}: Worst-case surface-to-volume ratios
  start from a plateau from where they descend in smooth steps into the
  quasi-optimal regime at $V_0$.
  Left: $d=2$ and $k=2$ (Hilbert curve).
  Right: $d=3$ and $k=3$ (Peano curve). 
  \label{fig:introduction:theorem}
   }
\end{figure}

Our paper's main contribution is Theorem \ref{th:surface_to_volume_asymptotic},
where we determine the maximum asymptotic surface-to-volume ratio for a given space-filling curve over all partitions of a given volume (Figure \ref{fig:introduction:theorem}):
 If we always refine towards the protruded corners of a partition, the surface
 scales linearly with the partition's cells, and the surface-to-volume ratio is
 bounded.
 Once these corners are saturated, i.e.~cannot accommodate further
 refinement, 
 we refine towards the edges,
 then the faces, and so forth.
 We obtain a smoothed step pattern, where the saturation points are labelled as
 $V_{d}$, $V_{d-1}$, \ldots. 
 Eventually, all of the partitions' boundaries are adaptively meshed yet the partitions themselves remain coarse
 inside.
 We ``fill'' their interior until we end up with a regular mesh and, hence, a
 quasi-optimal surface-to-volume ratio
 (Figure \ref{fig:introduction:theorem}).
 Even though our results are asymptotic, the rate of convergence is exponential with the depth of the mesh.
 To the best of our knowledge, no literature explicitly tackles adaptive grids
 as they are of particular relevance in scientific computing with its adaptive
 mesh refinement (AMR) or fast multipole algorithms, e.g.

To the best of our knowledge, the only attempt at evaluating the surface-to-volume
ratio of space-filling curve partitions for adaptive grids is due to
Zumbusch \cite[Chapter 4]{Zumbusch:2003:Habilitation},
where it is discussed that adaptive grids that refine aggressively towards a singularity yield high surface-to-volume ratio.
However, this work is limited to certain kinds of adaptive grids, and its main results (Lemmas 4.17 and 4.19) are limited to analogues of the upper bound in \eqref{equation:introduction:quasi-optimal} for those grids. 
In contrast, our work applies to all grids.
It yields quantitative expressions for the maximum surface-to-volume ratios and provides a theoretical mindset how to study and analyse SFC partitions.

The paper is organised as follows.
As the paper is technical in places, we kick off with an informal 
sketch of our overall proof in Section \ref{sec:sketch}. 
In Section \ref{sec:grids}, we review some necessary properties
of our grids and introduce our terminology. 
We also obtain preliminary results on space-filling
curves and we introduce the concept of a shape. 
From hereon, we define the discrete surface and volume and introduce the concept
of classified partitions (Section \ref{sec:classified_partitions}).
 The classic statement on quasi-optimal partitions for regular grids of large
 size is a direct consequence of the statements within this section
 (cmp.~Appendix \ref{sec:continuous-ratio}).
 Its main purpose however is to phrase which types of meshes have to be
 analysed, i.e.~which relations between cells and faces we find along the
 partitions' boundaries. 
 In Section \ref{sec:shape-class-regular-partitions}, we define these
 shape-class-regular partitions and derive their surface-to-volume
 ratios.
 We maximise this ratio in Section 
 \ref{sec:maximum-surface-to-volume-ratio}.
 In Section \ref{sec:conclusion}'s conclusion, we revise the paper's key
 insights, and use these application remarks to sketch how our insight can guide
 future usage of SFCs in codes.
We insert remarks on the work rationale and ideas
into the paper.
Further to that, we add digressions statements to the text that create links to  applications of space-filling curves in scientific computing.

\section{Sketch of proof strategy} 
\label{sec:sketch}

\begin{figure}
 \begin{center}
  \includegraphics[width=0.22\textwidth]{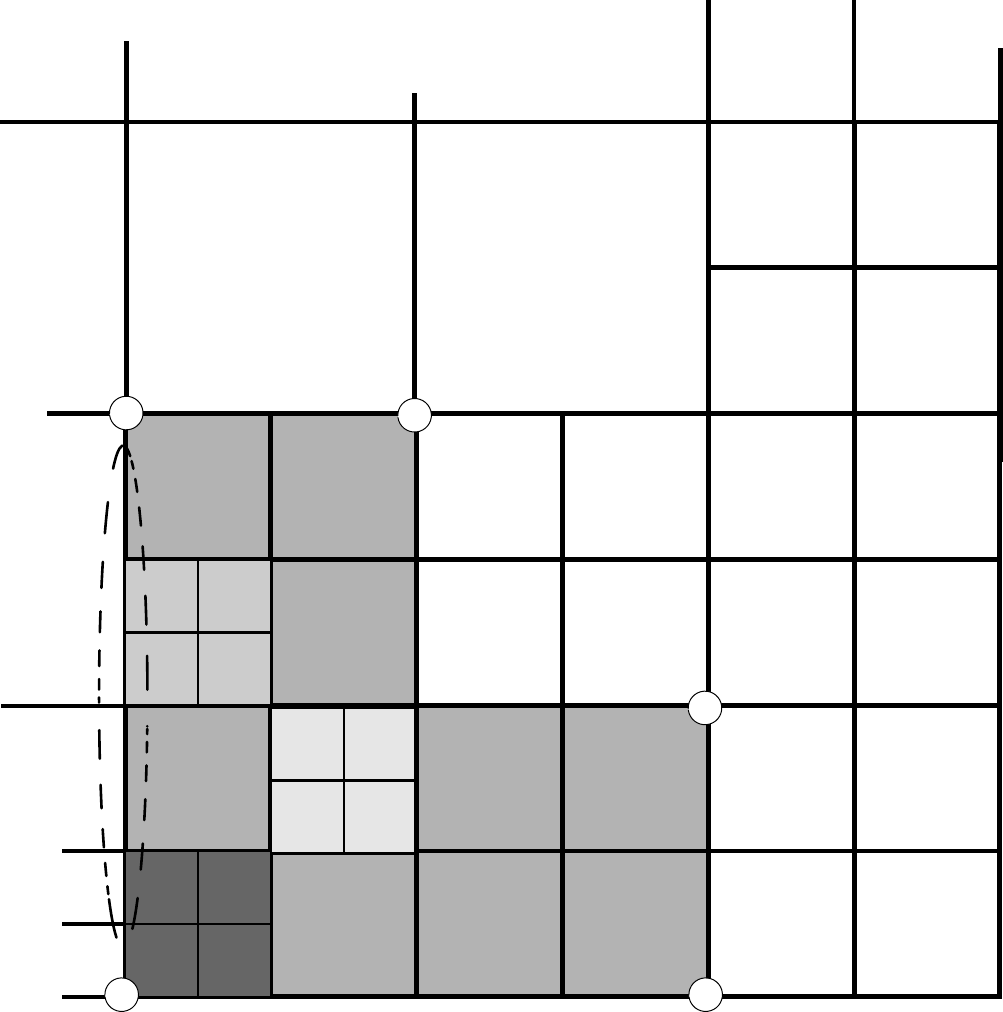}
 \end{center}
 \vspace{-0.1cm}
 \caption{
  An L-shaped partition with an adaptive quadtree mesh.
  Different refinements affect the ratio of boundary faces to inner cells
  differently.
  From dark to lighter:
  Refinement towards the five ``outer'' vertices yields the highest number of
  boundary faces compared to inner cells.
  Once a maximum mesh depth along the vertices is reached, refining along the edges of the
  partition (dotted circle area, e.g.) yields the worst-case ratio.
  Eventually, we refine the inner cells.
  \label{figure:terminology:grid-definition}
  \label{figure:sketch}
 }
\end{figure}

We work with grid partitions as they arise from spacetrees, i.e.~the  generalisation of the octree or quadtree concept
\cite{Weinzierl:2019:Peano}.
Our discussion first formalises all required language. 
Particular care is appropriate to clarify how we deal with partition boundaries
along refinement transitions, i.e.~regions where rather coarse mesh regions meet
fine regions.

With all terminology in place, we assume that we are given a partition, know its
ratio of inner cells to faces, and now are allowed to refine within the
partition.
Different partitions yield different updates to the ratio of boundary faces to
cells (Figure \ref{figure:sketch}).
If we refine towards outer vertices of the domain, we typically add four faces
(for $d=2$) in return for four cells.
In contrast, if we refine towards dented vertices or within a domain, we add
further inner cells yet to not introduce any new partition boundary faces.

Our strategy now reads as follows:
\begin{itemize}
  \item We introduce a classification of domains and their boundaries such that
  we formalise the notion of problematic and advantageous parts of the domain.
  \item We assume that a partition is given and that we may add further cells up
  to a given refinement level. Our strategy is first to refine only around those
  partition areas where the ratio of boundary faces to cells deteriorates. This
  is a worst-case refinement.
  \item We continue to refine along the next ``most critical'' area once all
  potential worst-case locations are occupied.
\end{itemize}

\noindent
Once we have expressions for this befilling strategy, we compute the limits for
high level depths.

In order to obtain our results, we construct a new framework for the study of adaptive grids. We notably introduce the notion of a classified partition. Intuitively, a partition $P$ is classified if every cell that contains a $(d-c)$-face of $P$ contributes to exactly $c$ facets of $P$. More details are given in Section \ref{sec:classified_partitions} as to why we only consider classified partitions. Our second contribution is about classifying a given partition. We show in Theorem \ref{th:classification} that for any partition $P$, one can construct a classified partition $P^*$ by performing $O(M)$ subdivisions of cells in $P$. Therefore, considering classified partitions is at no loss of generality if the volume grows superlinearly with the depth.

The impact of the curve $\Phi$ on the maximum surface-to-volume ratio is provided by its measure $\mu_c(\Phi)$, defined in Section \ref{sec:shape-class-regular-partitions}, which corresponds to some geometrical property of the curve. For $c < d$, the measure actually coincides with the maximum continuous volume of the $c$-boundary of a shape for $\Phi$. For $c=d$, the situation is more complicated, and the measure corresponds to the distribution of vertices a shape can have.



\section{Grids, shapes and space-filling curves} \label{sec:grids}

\begin{rationale}
Our grids of interest are described in a box language,
where the grid is derived from one big box by recursive, equidistant
$k$-subdivisions along the coordinate axes.
Such a recursive approach yields a cascade of adaptive Cartesian grids
of boxes.
For $d=2$ and $k=2$, the construction process equals the definition of 
quadtrees.
For $d=3$ and $k=2$, we obtain octrees.
Different to boxes exhibiting a tree structure, the resulting grid is a plain,
non-hierarchical object.
Our surface-to-volume statements refer to these grids yet exploit the underlying
box construction.

We partition the grids with space-filling curves (SFC), i.e.~we cut the
(sequence of) grid cells into chunks along the SFC ordering.
Any grid partition uniquely labels all boxes within the construction
hierarchy, too:
Either all descendants of a box are contained within a partition or not.
The coarsest boxes meeting the former criterion span the shape of a partition.
%
%
%
With properties shown for these shapes, we eventually return to the box language
to study ``what type of box refinement could we 
 embed into the shape
to construct grids with worst-case surface-to-volume ratios''.
\end{rationale}

\subsection{Boxes}

We use the standard notation for closed intervals of $\mathbb{R}$: $[a,b] = \{c \in \mathbb{R} : a \le c \le b\}$. We also introduce for any non-negative integer $n \in \mathbb{N}$ the notation $[n] = \{0,\dots,n-1\}$.

A \Define{hypercube} of dimension $d$ and size $\lambda$ is any subset of
$\mathbb{R}^d$ of the form $[a_1, b_1] \times \dots \times [a_d, b_d]$, where
$b_i - a_i = \lambda$ for all $1 \le i \le d$. 
The unit $d$-hypercube is $\hypercube := [0,1]^d$. 
We now introduce boxes for the hypercube $\hypercube$; these could easily be
extended to any hypercube. 
A \Define{box} is any hypercube $x \subseteq \hypercube$ of the form $[a_1, b_1]
\times \dots \times [a_d, b_d]$ of size $k^{-l}$ for some $k \ge 2$ and $l \ge
0$, with $a_i \in k^{-l} \Z$ for all $1 \le i \le d$. 
The integer $l$ is the \Define{depth} of the box $x$, which we shall denote as
$\depth(x)$.

We assume that the dimension $d$ of
the hypercube and the factor $k$ appearing in the size of boxes are fixed; as such, we shall omit them in our notation and terminology. We can represent the box $x$ more concisely by a word $(x_1, \dots, x_d) \in [k^l]^d$, where $x_i = a_i k^l$ (technically, we need the depth $l$ as well in order to make that encoding non-ambiguous; this will be clear from the context).

The set $\boxes$ of boxes of $\hypercube$ can be partially ordered with respect to inclusion. The Hasse diagram of that partial order is an infinite tree $\tree$, rooted at $\hypercube$, and where every vertex (box) has $k^d$ children. We denote the parent of the box $x$ as $\parent{x}$. We use the standard terminology of trees: if there is a path from the root to a box $x$ via another box $y$, then $y$ is an ancestor of $x$ and $x$ is a descendant of $y$. Note that in that case, $x \subseteq y$. We note that $\depth(x)$ is the length of the unique path from the root to $x$ in $\tree$. For any two boxes $x$ and $y$, their least common ancestor is denoted as $x \lor y$.

For any box $x$, let $C_x = \{y \in \boxes : x = \parent{y} \}$ be the set of
children boxes of $x$. The \Define{subdivision} of $x$ (or subdividing $x$)
corresponds to replacing $x$ by $C_x$, that is for any $X \subseteq \boxes$, $X$ is unchanged if $x \notin X$ and $X$ becomes $(X \setminus \{x\}) \cup C_x$ if $x \in X$.

In order to differentiate a set of boxes $X \subseteq \boxes$ with its actual
realisation in $\hypercube$, we define the \Define{content} of $X$ as
\[
	\content{X} := \bigcup_{x \in X} x \subseteq \hypercube. 
\]
Note that $\content{x} = x$ for any box $x$. For any sets of boxes $X$ and $Y$,
we say $Y$ is a \Define{refinement} of $X$ (or $Y$ \Define{refines} $X$) and we
denote $Y \refines X$ if $\content{X} = \content{Y}$ and if for any $y \in Y$, there exists $x \in X$ such that $y \subseteq x$. We note that the concept of subdivision is only for one box $v$: one replaces $v$ with all its children boxes. On the other hand, the concept of refinement is more general as a set of boxes $Y$ refines another set of boxes $X$ if one can obtain $Y$ by starting from $X$ and repeatedly subdividing boxes.

\begin{application}
Our box language yields non-conformal, structured adaptive Cartesian grids as
we know them from spacetrees as generalisation of octrees or quadtrees
\cite{Weinzierl:2019:Peano}.
This cell-based refinement is a degeneration of block-structured adaptive mesh
refinement which typically employs $k \gg 2$  \cite{Dubey:16:SAMR}.
\end{application}

\subsection{Grids}

It is clear that for two boxes $x$ and $y$, $x \cap y \in \boxes$ if and only if $x \subseteq y$ or $y \subseteq x$. We say a set of boxes $A$ is an \Define{antichain of boxes} if for any distinct $a,b \in A$, $a \not\subseteq b$; in other words, it is an antichain according to containment. The antichain of boxes $A$ is \Define{maximal} if there is no other antichain of boxes $B$ with $A \subset B$. 

A \Define{grid} is a finite maximal antichain of boxes. Grids have several alternative definitions, gathered in the following lemma. The proof is obvious and hence omitted.

\begin{lemma} \label{lem:grid}
Let $G$ be a finite set of boxes. Then the following are equivalent.
\begin{enumerate}[label = (\roman*)]
    \item $G$ is a grid.

    \item $\content{G} = \hypercube$, and $\content{H} \ne \hypercube$ for any $H \subset G$.

    \item $\content{G} = \hypercube$, and $g \not\subseteq g'$ for any distinct $g, g' \in G$.

    \item $G$ is obtained by successive subdivisions, starting from the one-cell grid $\{\hypercube\}$.
\end{enumerate}
\end{lemma}

\begin{application}
 Our mesh definition does not impose any (2:1) balancing 
 \cite{Clevenger:2020:MultigridFEM,Isaac:12:Balancing,Suh:2020:Balancing,Sundar:08:BalancedOctrees}.
 There is no constraint that two neighbouring cells'
 level may not differ by more than one.
 Our estimates thus are worst-case, as we
 allow the interior of a domain to become coarser
 rapidly.
 The inner cell count can be smaller than in a balanced tree.
 At the same time, we neglect the fact that fine-to-coarse transitions in the mesh may coincide with partition
 boundaries, i.e.~$k$ neighbouring cells within a partition might face one
 non-partition cell.
 This argument also holds the other way round, and it generalises---without 2:1
 balancing---over multiple levels.
 In such cases, many implementations would collect $k$ messages and exchange
 only once.
 This fact is neglected by our estimates.
\end{application}

 In line with our definition of subdivison which means replacing a cell with
 $k^d$ new ones, any set of boxes that results from an iterative subdivision of
 $\hypercube$ is, by definition, a grid.
And conversely, any grid can be obtained that way. A grid covers the whole domain.
 It does not host any overlaps of boxes, i.e.~if a box is in the grid, none of
 its children in the tree graph are members of the grid.

In order to avoid confusion, we refer to the boxes in $G$ as \Define{cells}, while any box that contains a cell in $G$ is referred to as a \Define{node} of $G$. The \Define{depth} of the grid $G$ is the maximum depth of a cell in $G$. A grid is \Define{regular} if all its cells have the same depth, thus the regular grid of depth $M$ has $k^{Md}$ cells. We summarise that a grid is a non-hierarchical, flat object as opposed to an arbitrary set of boxes or tree, respectively.

\begin{application}
 Our work studies ``flat'' grids 
 consisting of non-overlapping cells.
 We assume that all compute work is done on the  cells
 with the finest resolution (maximum depth) and that cells exchange information
 with their neighbours only.
 This flat notion breaks down for
 multiscale algorithms such as multigrid of fast multipole which compute different things on different
 resolution levels.
 However, our ``flat''
 estimates continue to offer valid estimates for their communication behaviour:
 The fixed subdivision ratio $k$ implies that the fine grid estimate establishes
 a natural upper bound on the face and cell counts on the next coarser
 resolution level via a scaling of $k^{-d}$ or $k^{-d+1}$, respectively.
 The argument carries over to coarser levels recursively.
 It ties in with the notation of a local essential tree (LET) in literature
 \cite{March:15:AlgebraicTreeCode}.
 While our multiscale extrapolation holds for data exchange cardinalities and
 thus bandwidth demands, e.g., it ignores latency effects which tend to
 gain importance for multiscale algorithms.
 The estimates also do not hold for codes which construct their coarse
 resolution representations algebraically, i.e.~not using the
 prescribed $k$-subdivision pattern.
\end{application}


\noindent
For any finite set of boxes $X \subseteq \boxes$, we define $\tree(X)$ as the minimal subtree of $\tree$ such that:
\begin{enumerate}
	\item $X$ is a subset of vertices of $\tree(X)$,
	
	\item if $v \in \tree(X)$ and $v \subseteq u$, then $u \in \tree(X)$,
	
	\item every non-leaf in $\tree(X)$ has exactly $k^d$ children.
\end{enumerate}
It is easily seen that every minimal element in $X$ (w.r.t. containment) is a leaf of $\tree(X)$.  
If $G$ is a grid, then $X \subseteq G$ (i.e.~$X$ is a subset of cells of $G$) if and only if $X$ is a subset of leaves of $\tree(G)$.

Let the \Define{minimal grid} of $X$ be the grid $\grid(X)$ satisfying
\[
    \tree(\grid(X)) = \tree(X).
\]
We remark that $\grid(X)$ is obtained by, beginning with the one-cell grid, repeatedly subdividing boxes that contain an element of $X$. Any grid $G$ that contains $X$ is obtained by refining cells of $\grid(X)$ outside of $X$. Finally, it is obvious that $G = \grid( G )$ for any grid $G$.

Clearly, for two grids $G$ and $H$, $G$ refines $H$ if and only if $\tree(H)$ is
a subgraph of $\tree(G)$. In general, the \Define{common refinement} of two grids
$G$ and $H$ is denoted as $G \land H$. It is the unique grid $T$ such that if $S \refines G$ and $S \refines H$, then $S \refines T$ for any grid $S$. In fact, for any finite set of grids $G_1, \dots, G_n$, another grid $S$ satisfies $S \refines G_i$ for all $1 \le i \le n$ if and only if $S \refines \bigwedge_{i=1}^n G_i$. It is easily shown that for any finite set of boxes $X$, $\grid(X) = \bigwedge_{x \in X} \grid(x)$.

\begin{application}
  Some 2:1 balancing algorithms clarify that a naive implementation of any
  balancing leads into a rippling: additional cells are added to mitigate some
  resolution transitions, but lead in turn to new violations of the balancing.
  An iterative approach thus might ripple changes through the domain whereas
  each iteration requires some parallel computations.
  A more sophisticated algorithm \cite{Suh:2020:Balancing}
  compresses the mesh by storing solely the
  finest grid cells, exchanges this compressed code and makes
  each partition fill coarser cells back in in a balanced way. Our minimal grids resemble the compressed mesh if only
  the first and last cell within a partition feed into the compression.
\end{application}

\subsection{Discrete space-filling curves}

We are interested in specific orderings $\Gamma = (g_1, \dots, g_n)$ of the cells of a grid $G$. 
Say the ordering is \Define{space-filling} if for all $1 \le i < j < k \le n$,
$g_i \lor g_j \subseteq g_i \lor g_k$. Two $d$-hypercubes are \Define{adjacent} or face-connected
if their intersection is a $(d-1)$-hypercube.

A \Define{Discrete Space-Filling Curve} (DSFC) $\Gamma$ on a grid $G$ is a total
ordering of the cells $\Gamma = (g_1, \dots, g_n)$ of $G$ such that the following hold.
\begin{enumerate}
	\item Continuity: for all $1 \le i < n$, $g_i$ and $g_{i+1}$ are adjacent;
	
	\item Space-filling: for all $1 \le i < j < k \le n$, $g_i \lor g_j \subseteq g_i \lor g_k$.
\end{enumerate}

\begin{rationale}
 Our definition over the box language resembles the formalism in
  \cite{Haverkort:2010:LocalityFor2DSFCs} who rightly points out that the
  terminology (discrete) space-filling curve is inferior to terms like
  ``scanning order'' which emphasise the ordering over volumes.
  The discussion in
  \cite{Bader:2013:SFCs} is even more rigorous, avoids the term curve---which
  is reserved for an object resulting from the limit over
  refinements---altogether, and
  refers to a ``space-filling order'' to distinguish it from the limit of 
  infinite refinement which eventually yields a curve
  \cite{Gotsman:1996:MetricsOnSFCs,Sagan:94:SFCs}.
  We stick to DSFC, but emphasise that our term does not make assumptions about some self-similarity
  \cite{Cardona:2016:SelfSimilarity,Schrack:2015:GenerateSFCs} and does not
  imply any volumetric homogeneity:

  We assume that the ordered volumes are non-overlapping
  \cite{Cardona:2016:SelfSimilarity} and face-connected.
  In line with \cite{Haverkort:2010:LocalityFor2DSFCs},
  the important generalisation compared to other work in the field is that our
  DSFC defines an order over adaptive Cartesian meshes, i.e.~the ordered
  hypercubes can have different size.
  Our term DSFC thus differs from the classic notion of a (curve) ``iterate'' or
  successive production through a grammar, where we typically assume a uniform
  unfolding of the underlying tree structure
  \cite{Cardona:2016:SelfSimilarity,Gotsman:1996:MetricsOnSFCs,Hungershoefer:02:SFCQuality,Schrack:2015:GenerateSFCs,Zumbusch:2001:QualityOfSFCs}.
 
\end{rationale}

\noindent
The space-filling property has two alternate
definitions:

\begin{lemma} \label{lemma:space-filling}
Let $G$ be a grid. For any ordering $\Gamma = (g_1, \dots, g_n)$ of the cells of $G$, the following are equivalent.
\begin{enumerate}[label=(\roman*)]
	\item \label{item:space-filling1} 
	For all $1 \le i < j < k \le n$, $g_i \lor g_j \subseteq g_i \lor g_k$.

	\item \label{item:space-filling2}
	For all $1 \le i < j < k \le n$, $g_j \lor g_k \subseteq g_i \lor g_k$.

	\item \label{item:space-filling3}
	For any $1 \le \alpha \le \beta \le n$, 
	\[
		g_\alpha \lor g_\beta = \bigvee_{i=\alpha}^\beta g_i.
	\]
\end{enumerate}
\end{lemma}

\begin{proof}
Clearly, \ref{item:space-filling3} implies both \ref{item:space-filling1} and \ref{item:space-filling2}. We now prove \ref{item:space-filling1} implies \ref{item:space-filling3}. Since $g_\alpha \lor g_\beta \subseteq \bigvee_{i=\alpha}^\beta g_i$, we only need to prove the reverse inclusion. We prove this by induction on $\delta := \beta - \alpha$. This is clear for $\delta = 0$, so suppose for up to $\delta - 1$. We have
\[
	\bigvee_{i=\alpha}^\beta g_i = \bigvee_{i=\alpha}^{\beta-1} g_i \lor g_\beta = (g_\alpha \lor g_{\beta - 1}) \lor g_\beta \subseteq (g_\alpha \lor g_\beta) \lor g_\beta = g_\alpha \lor g_\beta.
\]
The proof that \ref{item:space-filling2} implies \ref{item:space-filling3} is similar and hence omitted.
\end{proof}

An SFC \Define{partition} $P$ (of a DSFC $\Gamma$) on a grid $G$ is a set of consecutive
cells along the discrete space-filling curve $\Gamma$:
$P = \{g_i, \dots, g_j\}$ for some $1 \le i \le j \le n$.

We denote the ordering of the cells as $g_i \le_\Gamma g_j$ for any $i \le j$. For any node $v$, the cells contained in $v$ are consecutive according to $\Gamma$ (otherwise this would contradict the space-filling property). More generally, a DSFC then induces a total ordering of the nodes of $G$: say $u \le_\Gamma v$ if either $u \subseteq v$ or there exist two cells $x \subseteq u$ and $y \subseteq v$ such that $x \le_\Gamma y$ (or equivalently, $x' \le_\Gamma y'$ for all $x' \subseteq u$ and $y' \subseteq v$). The space-filling property is given in its most general form as follows.

\begin{theorem}
Let $\Gamma$ be a DSFC on a grid $G$. Then for any two nodes $u$ and $v$ of $G$, 
\[
	u \lor v = \bigvee_{u \le_\Gamma x \le_\Gamma v} x.
\]
\end{theorem}

\begin{proof}
Similarly to Lemma \ref{lemma:space-filling}, the result is equivalent to: let $u, v, w$ be nodes of $G$ with $u \le_\Gamma v \le_\Gamma w$, then $u \lor v \subseteq u \lor w$.

The result is clear if $v \subseteq w$ so we assume $v \not\subseteq w$. Let us first suppose that $u \subseteq v$, so that $u \lor v = v$. Then $z := u \lor w$ is a node that intersects $v$ nontrivially (since $u \subseteq z \cap v$), then either $z \subset v$ or $v \subseteq z$. Since $z$ has a part outside of $v$ (because $w \not\subseteq v$), we must have $v \subseteq z$.

Let us now suppose that $u \not\subseteq v$. We can denote the cells belonging to $u$, $v$ and $w$ as $g_\alpha, \dots, g_\beta$, $g_\gamma, \dots, g_\delta$, and $g_\epsilon, \dots, g_\zeta$ respectively for some $\alpha \le \beta < \gamma \le \delta < \epsilon \le \zeta$. By Lemma \ref{lemma:space-filling}, we obtain
\begin{align*}
	u \lor v &= g_\alpha \lor g_\delta,\\
	u \lor w &= g_\alpha \lor g_\zeta.
\end{align*}
The space-filling property then yields $u \lor v \subseteq u \lor w$.
\end{proof}

\noindent
A \Define{space-filling curve} $\Phi$ is a function that associates a DSFC $\Phi(G)$ to every grid
$G$, and that preserves the ordering of nodes. More formally, if $u$ and $v$ are nodes of $G$ with $u \le_{\Phi(G)} v$ and $G'$ refines $G$, then $u \le_{\Phi(G')} v$.


 The actual (continuous) curve then results from the limit for infinite
 subdivision.
 To meet the continuity requirement, SFCs rely on appropriate rotation and
 mirroring of a leitmotif per refinement step.
 This distinguishes the Peano from the Hilbert from the Lebesgue curve.
 The latter weakens the continuity requirements and phrases it in terms of a  
 Cantor set as preimage \cite{Bader:2013:SFCs}.
 As the continuous space-filling curve is of no further interest in this work,
 we neglect further continuous or limit properties and even use the term SFC as
 synonym for DSFC.
 To simply our work, we furthermore stick to continuous SFCs in the above sense.
 The extension of our statements to Lebesgue (z-ordering) is straightforward, as
 the number of discontinuous subdomains produced by this SFC is bounded
 (cmp.~\cite{Burstedde:19:ZCurve}).
 It is also possible to extend out work to non-cubic box
 hierarchies (Sierpinksi) or weak derivations of SFCs
 \cite{Burstedde:2016:Tetrahedra}.




%
%

\subsection{Shapes}

A \Define{decomposition} of a set of boxes $X$ is another set of boxes $S$ such
that $\content{X} = \content{S}$. For any finite set of boxes $X$, let $Q$ be the set of maximal boxes in $\content{X}$, i.e.
\[
	Q = \{ q \subseteq \content{X} : q \subset y \implies y \not\subseteq \content{X} \}.
\]
We refer to $Q$ as the \Define{shape} of $X$, and we denote it as $\shape(X)$.

\begin{lemma} \label{lem:shape_is_decomposition}
For any finite set of boxes $X$, $\shape(X)$ is the unique decomposition of $X$ of minimum cardinality.
\end{lemma}

\begin{proof}
Firstly, we prove that $Q := \shape(X)$ is a decomposition of $X$, i.e. that $\content{Q} = \content{X}$. Let $Z = \{ z \in \boxes : z \subseteq \content{X} \}$; it is clear that $\content{Z} = \content{X}$. Now, since $Q \subseteq Z$, we have $\content{Q} \subseteq \content{Z} = \content{X}$. Conversely, for every $z \in Z$, there exists $q \in Q$ such that $z \subseteq q$, thus $\content{Q} = \content{Z} = \content{X}$.

Secondly, we prove that $Q$ is the unique decomposition of minimum cardinality.
Let $S$ be a decomposition of $X$ with minimum cardinality; we shall prove that $S = Q$. Note that $S \subseteq Z$, thus for every $s \in S$, there exists $q_s \in Q$ such that $s \subseteq q_s$. We obtain
\[
	\content{X} = \content{S} = \bigcup_{s \in S} s \subseteq \bigcup_{s \in S} q_s \subseteq \content{Q} = \content{X}.
\]
Therefore, we have $|S| \ge |Q|$ (otherwise, the second inclusion would be strict) and hence $|S| = |Q|$. We also have $s = q_s$ for all $s$ (otherwise the first inclusion would be strict) and hence $S = Q$.
\end{proof}

We naturally say a set of boxes $Q$ is a shape if $Q = \shape(X)$ for some finite set of boxes $X$. Shapes can be characterised as follows.

\begin{lemma}
The following are equivalent:
\begin{enumerate}[label=(\roman*)]
	\item \label{it:shape_def}
	$Q$ is a shape.

	\item \label{it:shape_R}
	For any $R \subseteq Q$, $\content{R}$ is a box if and only if $|R| = 1$.
	
	\item \label{it:shape_v}
	For any node $v \subseteq \content{Q}$, there is a unique $q \in Q$ such that $v \subseteq q$.
	
	\item \label{it:shape_Q} 
	$Q = \shape(Q)$.
\end{enumerate}
\end{lemma}

\begin{proof}
\ref{it:shape_def} $\implies$ \ref{it:shape_R}. Let $Q = \shape(X)$, and let $Z$ be the set of boxes contained in $\content{X}$. Let $R \subseteq Q$, $|R| \ge 2$, such that $\content{R}$ is a box. Say $r \in R$, then $r$ is maximal in $Z$ by definition of $Q$, but it is not maximal in $Z$ since it is strictly contained in $\content{R}$, which is the desired contradiction.

\ref{it:shape_R} $\implies$ \ref{it:shape_v}. Let $v \subseteq \content{Q}$. Firstly, suppose that there is no $q \in Q$ containing $v$. Then $v = \content{R}$ for some $R \subseteq Q$, $|R| \ge 2$, which contradicts the hypothesis. Secondly, suppose that $v \subseteq q_1$ and $v \subseteq q_2$ for some distinct $q_1, q_2 \in Q$. Then $q_1 \subseteq q_2$ (or vice versa), and hence $\content{ q_1, q_2 } = q_2$, which once again contradicts the hypothesis.

\ref{it:shape_v} $\implies$ \ref{it:shape_Q}. Let $Y = \shape(Q)$, then for any $y \in Y$, let $q_y$ be the unique box in $Q$ that contains $y$. Since $y$ is maximal, we have $q_y = y$ and hence $Y \subseteq Q$. If the inequality is strict, then $\content{Y} \subset \content{Q}$, which contradicts Lemma \ref{lem:shape_is_decomposition}.

\ref{it:shape_Q} $\implies$ \ref{it:shape_def}. Trivial.
\end{proof}

\begin{application}
 The partitioning with (discrete) SFCs and the subsequent distribution of these
 chunks over compute units is a non-hierarchical technique which
 leaves the question open to which units the coarser grid levels (boxes) belong
 to.
 Two approaches are found in simulations
 \cite{Weinzierl:2019:Peano}, yet both make a box where all children boxes belonging to a partition $P$ belong to
 $P$, too.
 Codes that replicate coarser levels replicate a coarse box whose children are
 members of $P_1, P_2, \ldots$ on all compute units handling these $P_1, P_2,
 \ldots$.
 This is a bottom-up approach where the root box is
 shared among all compute units.
 A top-down approach assigns the coarse box in question uniquely to one of the
 compute units responsible for a child box.
 With the first-child rule \cite{Clevenger:2020:MultigridFEM}, the first child
 of a box along the DSFC determines which compute unit is responsible for the
 coarser box.
 Both paradigms, bottom-up and top-down, are different to our shape formalism: 
 It explicitly excludes boxes where the compute unit assignments would not be
 unique from the analysis.
\end{application}


%
%
%

\subsection{Continuous and discrete volume and surface}

\begin{rationale}
 Our paper orbits around the analysis of face-cuts:
  When we split up a grid into partitions, we cut ``only'' along faces.
  We furthermore assume that algorithms working on cells make cells exchange
  information through their interfaces.
  The number of cells we cut through
  vs.~the number of cells we hold in one partition thus is a reasonable
  characterisation for an algorithm's parallel behaviour. 
  To understand the interplay of face-cuts and mesh refinement,
  it is
 important to distinguish refinement along faces from edges from vertices: 
 If AMR refines towards an edge, we create, relative to the additional $k^d$
 inner cells per cell refinement, more boundary faces than with a
 refinement towards a face.
 We therefore first classify vertex, edge, face, \ldots for an isolated cell,
 and then generalise this concept to the boundary of a partition.
\end{rationale}



A \Define{$c$-subcube} of a hypercube is one of its $(d-c)$-dimensional hypercube faces.
For a given hypercube $x = [a_1, b_1] \times \dots \times [a_d, b_d]$, it results
from picking $c$ dimensions and fixing the corresponding coordinate entry
$x_i \in \{ a_i, b_i \}$. In other words, it is any hypercube of the form
\[
 	x_{S,T} = \{ z \in x : z_i = a_i \,\forall i \in S, z_j = b_j \,\forall j \in T \}
 \] 
for some $S, T \subseteq [d]$, $S \cap T = \emptyset$ and $|S| + |T| = d-c$.

We denote the set of $c$-subcubes of $x$ as $\subcube^c x$, so that $\subcube^0 x = \{x\}$. We also denote $\subcube x = \subcube^1 x$ and $\subcube^c G := \bigcup_{g' \in G} \subcube^c g' :  s \subseteq x \in \subcube^c g$.

For any set of cells $X$ of a grid $G$ and any $g \in X$, we define the $c$-\Define{boundary} of $g$ with respect to $X$ as
\[
	\boundary^c_X g := \left\{ s \in \subcube^c G, s \not\subseteq h \text{ for all } h \in X \setminus g \right\}.
\]
These are the $(d-c)$-subcubes of $G$ that belong to $g$ but to no other element of $X$. We extend the notation to $\boundary^c_X Y = \bigcup_{g \in Y} \boundary^c_X g$ for all $Y \subseteq X$. In order to simplify notation, we denote $\boundary^c X = \boundary^c_X X$. In particular, we denote $\boundary X  := \boundary^1 X$. Any element of $\boundary^c X$ is called a $(d-c)$-\Define{face} of $X$, a $0$-face is a \Define{vertex} of $X$, while a $(d-1)$-face is a \Define{facet} of $X$. 

\begin{example} \label{example:P1}
 We shall illustrate some of the key concepts in this paper with a running example of the Hilbert partition ($k=2$, $d=2$) given in
 Figure~\ref{fig:P}. 
 It consists of only three cells of different depths: $P = \{a, b, c\}$. 
 The subcubes $\subcube^c g$ and the contribution $\boundary_P^c g$ to the boundary of $P$ are given for each $g \in P$ in the table below:
~\\

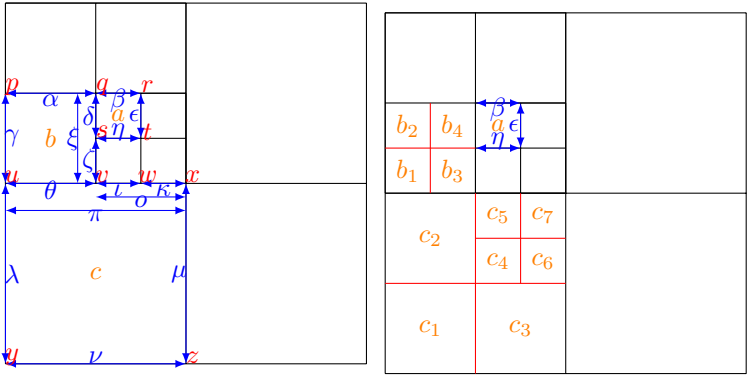
\begin{figure}
\centering
 \begin{tikzpicture}[scale=0.6]
\begin{scope}[scale=4]
    \draw (0,0) grid (2,2);
\end{scope}

\begin{scope}[scale=2]
    \draw (0,2) grid (2,4);
\end{scope}

\begin{scope}[scale=1]
    \draw (2,4) grid (4,6);
\end{scope}

\begin{scope}[color=orange]
\draw (2.5,5.5) node {$a$};
\draw (1,5) node {$b$};
\draw (2,2) node {$c$};
\end{scope}

\begin{scope}[color=blue]
\draw  (1,5.85) node {$\alpha$};
\draw[color=blue, latex-latex] (0,6) -- (2,6);

\draw  (2.5,5.85) node {$\beta$};
\draw[color=blue, latex-latex] (2,6) -- (3,6);

\draw  (0.15,5) node {$\gamma$};
\draw[color=blue, latex-latex] (0,4) -- (0,6);

\draw  (1.85,5.5) node {$\delta$};
\draw[color=blue, latex-latex] (2,5) -- (2,6);

\draw  (2.85,5.5)  node {$\epsilon$};
\draw[color=blue, latex-latex] (3,5) -- (3,6);

\draw  (1.85,4.5)  node {$\zeta$};
\draw[color=blue, latex-latex] (2,4) -- (2,5);

\draw  (2.5,5.15)  node {$\eta$};
\draw[color=blue, latex-latex] (2,5) -- (3,5);

\draw  (1,3.85)  node {$\theta$};
\draw[color=blue, latex-latex] (0,4) -- (2,4);

\draw  (2.5,3.85)  node {$\iota$};
\draw[color=blue, latex-latex] (2,4) -- (3,4);

\draw  (3.5,3.85)  node {$\kappa$};
\draw[color=blue, latex-latex] (3,4) -- (4,4);

\draw  (0.15,2)  node {$\lambda$};
\draw[color=blue, latex-latex] (0,0) -- (0,4);

\draw  (3.85,2)  node {$\mu$};
\draw[color=blue, latex-latex] (4,0) -- (4,4);

\draw  (2,0.15)  node {$\nu$};
\draw[color=blue, latex-latex] (0,0) -- (4,0);

\draw (1.5, 5) node {$\xi$};
\draw[latex-latex] (1.6,4) --  (1.6,6); 

\draw (3,3.6) node {$o$};
\draw[latex-latex] (2,3.7) -- (4,3.7); 

\draw (2,3.3) node {$\pi$};
\draw[latex-latex] (0,3.4) -- (4,3.4); 
\end{scope}

\begin{scope}[color=red]
\draw (0.15,6.15) node {$p$};
\draw (2.15,6.15) node {$q$};
\draw (3.15,6.15) node {$r$};
\draw (2.15,5.15) node {$s$};
\draw (3.15,5.15) node {$t$};
\draw (0.15,4.15) node {$u$};
\draw (2.15,4.15) node {$v$};
\draw (3.15,4.15) node {$w$};
\draw (4.15,4.15) node {$x$};
\draw (0.15,0.15) node {$y$};
\draw (4.15,0.15) node {$z$};
\end{scope}
\end{tikzpicture}
 \begin{tikzpicture}[scale=0.6]
\begin{scope}[scale=4]
    \draw (0,0) grid (2,2);
\end{scope}

\begin{scope}[scale=2]
    \draw (0,2) grid (2,4);
\end{scope}

\begin{scope}[scale=1]
    \draw (2,4) grid (4,6);
\end{scope}

\begin{scope}[color=red]
\draw (0,2) -- (4,2);
\draw (2,0) -- (2,4);

\draw (0,5) -- (2,5);
\draw (1,4) -- (1,6);

\draw (3,2) -- (3,4);
\draw (2,3) -- (4,3);
\end{scope}

\begin{scope}[color=orange]
\draw (2.5,5.5) node {$a$};

\draw (0.5,4.5) node {$b_1$};
\draw (0.5,5.5) node {$b_2$};
\draw (1.5,4.5) node {$b_3$};
\draw (1.5,5.5) node {$b_4$};

\draw (1,1) node {$c_1$};
\draw (1,3) node {$c_2$};
\draw (3,1) node {$c_3$};
\draw (2.5,2.5) node {$c_4$};
\draw (2.5,3.5) node {$c_5$};
\draw (3.5,2.5) node {$c_6$};
\draw (3.5,3.5) node {$c_7$};
\end{scope}

\begin{scope}[color=blue]

\draw  (2.5,5.85) node {$\beta$};
\draw[color=blue, latex-latex] (2,6) -- (3,6);



\draw  (2.85,5.5)  node {$\epsilon$};
\draw[color=blue, latex-latex] (3,5) -- (3,6);


\draw  (2.5,5.15)  node {$\eta$};
\draw[color=blue, latex-latex] (2,5) -- (3,5);









\end{scope}
\end{tikzpicture}
\caption{
 Left: The partition $P = \{a,b,c\}$ as
 discussed in Example \ref{example:P1}. Right:
 The pre-classification $\preclassified{P}$ of the same partition $P$.
 \label{figure:preclassifiedP}
 \label{fig:P}
}
\end{figure}

\begin{center}
\begin{tabular}{c|c|c|c|c|c|c}
    Cell $g$ & $\subcube^0 g$ & $\subcube^1 g$& $\subcube^2 g$& $\boundary^0_P g$& $\boundary^1_P g$& $\boundary^2_P g$\\
    \hline
    $a$ & $\{ a \}$ & $\{ \beta, \delta, \epsilon, \iota \}$ & $\{ q, r, s, t \}$ & $\{a\}$ & $\{ \beta, \epsilon, \iota \}$ & $\{ r, t \}$ \\
    $b$ & $\{ b \}$ & $\{ \alpha, \gamma, \xi, \theta \}$ & $\{ p, q, u, v \}$ & $\{b\}$ & $\{ \alpha, \gamma, \zeta \}$ & $\{ p \}$ \\
    $c$ & $\{ c \}$ & $\{ \pi, \lambda, \mu, \nu \}$ & $\{ u, x, y, z \}$ & $\{c\}$ & $\{ \iota, \kappa, \lambda, \mu, \nu \}$ & $\{ x, y, z \}$
\end{tabular}
\end{center}
~\\

Thus for $e \in \{0,1,2\}$, the $e$-boundary of $P$ is given by
\begin{align*}
    \boundary^0 P &= \{ a, b, c \},\\
    \boundary^1 P &= \{ \alpha, \beta, \gamma, \epsilon, \eta, \zeta, \iota, \kappa, \lambda, \mu, \nu \},\\
    \boundary^2 P &= \{ p, r, t, x, y, z \}.
\end{align*}
\end{example}


For any dimension $e \ge 1$, the continuous volume of any $e$-hypercube of size $\lambda$ is $\lambda^e$. In particular, the continuous volume of a cell $x$ of a grid is $k^{-d \cdot \depth(x)}$. For any set $X$ of cells of a grid $G$, the \Define{continuous volume} and \Define{continuous surface} of $X$ follow its definition as a union of hypercubes:
\begin{align*}
	\cv(X) &:= \sum_{x \in X} k^{- d \cdot \depth(x)},\\
	\cs(X) &:= \cv(\boundary X).
\end{align*}
We remark that the continuous volume of $X$ does not depend on the actual grid. Moreover,Lemma \ref{lemma:cvQ} below shows that the continuous volume of the $c$-boundary of $X$ only depends on the shape (or equivalently, content) of $X$.

\begin{lemma} \label{lemma:cvQ}
For any finite set of boxes $X$ of shape $Q$, $\cv( \boundary^c X) = \cv( \boundary^c Q)$ for all $0 \le c \le d$.
\end{lemma}

\begin{proof}
We only need to prove that if $X$ is obtained from $Y$ by subdividing one cell of $Y$, then $\cv( \boundary^c X ) = \cv( \boundary^c Y )$. Suppose $X$ is obtained by subdividing $h \in Y$, with set of children $H \subseteq X$. Then 
\[
    \cv( \boundary^c Y ) = \cv( \boundary^c_Y h ) + \cv( \boundary^c_Y (Y \setminus h) ) = \cv( \boundary^c_X H ) + \cv( \boundary^c_X (X \setminus H) ) = \cv( \boundary^c X).
\]
\end{proof}

The \Define{discrete volume} of $X$ and the \Define{discrete surface} of $X$ with respect to $G$ are as follows:
\begin{align*}
	\dv(X) &:= |X|,\\
	\ds(X,G) &:= |\boundary X|.
\end{align*}
We then define the \Define{discrete surface} of $X$ as 
\begin{equation} \label{equation:ds}
    \ds(X) := \ds( X, \grid(X) ) = \min \{ \ds(X, G) : X \subseteq G \}.
\end{equation}
Considering the grid $\grid(X)$ makes the definition intrinsic: considering grids where the outside of $\content{X}$ is much more finely refined artificially increases the surface, as can be seen in Figure \ref{fig:minimal_grid}. The main topic of this paper is the study of the (discrete) \Define{surface-to-volume ratio} of $X$:
\[
    \dr(X) := \frac{ \ds(X) }{ \dv(X) }.
\]

As an illustration of the discrete and continuous volumes and surfaces, we consider the case where $X = \{g\}$ is a single cell of maximum depth $M$ in a grid $G$. We then have $\boundary^c g = \subcube^c g$ for all $0 \le c \le d$ as no facet of $g$ is subdivided. A simple counting argument then yields
\begin{equation} \label{equation:dvg}
    \dv(\boundary^c g) = \binom{d}{c} 2^c.
\end{equation}
In particular, $\dv(g) = 1$, $\ds(g) = 2d$ and hence $\dr(g) = 2d$.

Let $0 \le c \le d-1$. Since any element $s \in \subcube^c g$ is a $(d-c)$-hypercube of size $k^{-M}$, we have $\cv(s) = k^{-M(d-c)}$ whence
\begin{equation} \label{equation:cvg}
    \cv(\boundary^c g) = \binom{d}{c} 2^c k^{-M(d-c)}.
\end{equation}
In particular, $\cv(g) = k^{-Md}$, $\cs(g) = 2d k^{-M(d-1)}$. It is worth noting that, in view of Lemma \ref{lemma:cvQ}, Equation \eqref{equation:cvg} actually holds for any cell of $G$; on the other hand, Equation \eqref{equation:dvg} does not hold for all cells of $G$ in general.

\begin{application}
 We may read our formalism with the minimum in (\ref{equation:ds}) as a
 superregularisation along the partition boundaries where we forbid a partitioning to cut through the mesh along a mesh
 refinement.
 As such a regularisation is barely found in codes, our estimates are off by
 up to a  factor of $k$ 
  compared to a 2:1-balanced tree or, otherwise, $k^{\ell }$ with $\ell $ being
  the maximum mesh level transition.
\end{application}

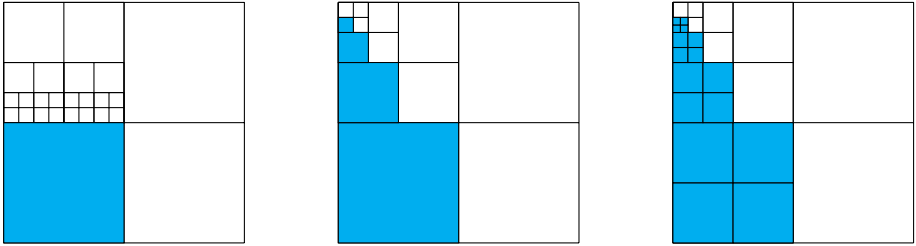
\begin{figure}[htp]
\centering
\begin{tikzpicture}[scale=0.4]
	\draw[step = 4cm] (0,0) grid (8,8);

	\draw[step = 2cm] (0,4) grid (4,8);

	\draw[step = 1cm] (0,4) grid (2,6);
	\draw[step = 1cm] (2,4) grid (4,6);

	\draw[step = 0.5cm] (0,4) grid (1,5);
	\draw[step = 0.5cm] (1,4) grid (2,5);
	\draw[step = 0.5cm] (2,4) grid (3,5);
	\draw[step = 0.5cm] (3,4) grid (4,5);
	
	\filldraw[fill=cyan, draw=black] (0,0) rectangle (4,4);
\end{tikzpicture}
\hspace{1cm}
\begin{tikzpicture}[scale=0.4]
	\draw[step = 4cm] (0,0) grid (8,8);
	\draw[step = 2cm] (0,4) grid (4,8);
	\draw[step = 1cm] (0,6) grid (2,8);
	\draw[step = 0.5cm] (0,7) grid (1,8);

	\filldraw[fill=cyan, draw=black] (0,0) rectangle (4,4);
	\filldraw[fill=cyan, draw=black] (0,4) rectangle (2,6);
	\filldraw[fill=cyan, draw=black] (0,6) rectangle (1,7);
	\filldraw[fill=cyan, draw=black] (0,7) rectangle (0.5,7.5);
\end{tikzpicture}
\hspace{1cm}
\begin{tikzpicture}[scale=0.4]
	\draw[step = 4cm] (0,0) grid (8,8);
	\draw[step = 2cm] (0,0) grid (4,8);
	\draw[step = 1cm] (0,4) grid (2,8);
	\draw[step = 0.5cm] (0,6) grid (1,8);
    \draw[step = 0.25cm] (0,7) grid (0.5,7.5);

	\filldraw[fill=cyan, draw=black] (0,0) rectangle (2,2);
	\filldraw[fill=cyan, draw=black] (0,2) rectangle (2,4);
	\filldraw[fill=cyan, draw=black] (2,0) rectangle (4,2);
	\filldraw[fill=cyan, draw=black] (2,2) rectangle (4,4);
	\filldraw[fill=cyan, draw=black] (0,4) rectangle (1,5);
	\filldraw[fill=cyan, draw=black] (1,4) rectangle (2,5);
	\filldraw[fill=cyan, draw=black] (0,5) rectangle (1,6);
	\filldraw[fill=cyan, draw=black] (1,5) rectangle (2,6);
	\filldraw[fill=cyan, draw=black] (0,6) rectangle (0.5,6.5);
	\filldraw[fill=cyan, draw=black] (0.5,6) rectangle (1,6.5);
	\filldraw[fill=cyan, draw=black] (0,6.5) rectangle (0.5,7);
	\filldraw[fill=cyan, draw=black] (0.5,6.5) rectangle (1,7);
	\filldraw[fill=cyan, draw=black] (0,7) rectangle (0.25,7.25);
	\filldraw[fill=cyan, draw=black] (0.25,7) rectangle (0.5,7.25);
	\filldraw[fill=cyan, draw=black] (0,7.25) rectangle (0.25,7.5);
	\filldraw[fill=cyan, draw=black] (0.25,7.25) rectangle (0.5,7.5);
\end{tikzpicture}

 \caption{
   Left:
   A degenerated partition with only one cell. Without further assumptions about
   the resolution transitions along the boundary, we can construct setups with
   an arbitrary high surface-to-volume ratio. Our ``balancing'' condition along
   the boundary is a cousin to the well-known 2:1-balancing, which we 
   do not enforce globally, i.e.~within partitions.
   Middle:
   Partition with linear number of non-classified cells as formalised in
   Section \ref{sec:classified_partitions}.
   Right:
   One further refinement yields a classified partition.
   \label{fig:minimal_grid}
 }
\end{figure}

\section{Classification} 
\label{sec:classified_partitions}

\subsection{Pre-classified partitions}

\begin{rationale}
 With the terminology in place, we next establish a language over cells that
 tells us how many partition boundary faces are induced by a single cell. We
 \emph{classify} cells by how much they contribute towards a parallel boundary
 and, hence, parallel data exchange.
 This per-cell nomenclature yields a nomenclature for whole partitions.
 Once we know how many boundary faces are introduced by a cell or would be 
 introduced by a further cell if we refined, 
 our overall strategy is to refine always cells with the worst-case face-to-cell
 ratio, i.e.~with the highest class.
 The classification provides us with an abstract language to formalise which
 cells do induce these worst cases.
 We conclude the section with some elementary properties of these cells.
 They later are plugged into the overall algorithm analysis. 
\end{rationale}

Let $P$ be a partition and $g \in P$. The class of $g$ with respect to $P$ is
\[
    \class(g) = \class( g, P ) := \max \{ c : \boundary^c_P g \ne \emptyset \}.
\]
It is the maximum $c$ such that $g$ contains a $(d-c)$-face of $P$: in particular $\class(g) = d$ if $g$ contains a vertex of $P$, while $\class(g) = 0$ if $g$ is in the interior of $P$.

\begin{example} \label{example:P2}
We continue Example \ref{example:P1} by illustrating the concept of class. All three cells in $P$ have class $2$, since $r, p, x \in \boundary^2 P$ belong to $a$, $b$, and $c$ respectively.
\end{example}

A cell $g \in P$ is \Define{pre-classified} (w.r.t. $P$) if for every $p \in \boundary P$, $p \subseteq g$ implies $p \in \subcube g$. In other words, $g$ is pre-classified if no
$(d-1)$-subcube of $g$ on the boundary of $P$ is subdivided. 
A partition is pre-classified if all its cells are pre-classified. It is clear that pre-classified partitions are closed under refinement.

 Our pre-classification implies that the boundary of a partition is not
 subdivided ``on the other side'', i.e.~no adjacent partition subdivides a
 boundary further.
 Therefore, studies of a partition study its worst-case:
 Adjacent partitions might work with coarser grids along the joint boundary and
 thus see a smaller face count, but they never work with a finer resolution
 along $\boundary P$.
 This characterisation is unidirectional.
 If a boundary border for one partition fits to the definition of
 pre-classified, this does not imply that it suits the definition for an
 adjacent partition.

We now introduce the \Define{pre-classification} of $P$, denoted as $\preclassified{P}$. First of all, for any $y \in \boundary P$, let $\hat{y} \subseteq \content{P}$ be the box such that $y \in \subcube \hat{y}$. We then denote $\hat{\boundary P} = \{ \hat{y} : y \in \boundary P \}$. Then the pre-classification of $P$ is obtained from $P$ by refining towards $\boundary P$:
\[
    \preclassified{P} := \grid( \hat{\boundary P} )[P].
\]
The following lemma justifies our terminology.

\begin{lemma} \label{lemma:preclassification}
Let $P$ be a partition, then the following hold.
\begin{enumerate}[label=(\alph*)]
    \item $\preclassified{P}$ is pre-classified.
    
    \item $\depth(\preclassified{P}) = \depth(P)$.
    
    \item A partition is a pre-classified refinement of $P$ if and only if it is a refinement of $\preclassified{P}$.
\end{enumerate}
\end{lemma}

\begin{proof}
~
\begin{enumerate}[label=(\alph*)]
    \item Let $g \in \preclassified{P}$. By construction, for any $y' \in \subcube g$ with $y' \subseteq \content{ \boundary P }$, there exists $y \in \boundary P$ such that $y' \subseteq y$. Therefore $g \subseteq \hat{y}$, from which $y' = g \cap y \in \boundary \preclassified{P}$ and $g = \hat{y'}$. Thus, no subcube of $g$ on the boundary of $P$ is subdivided, and $g$ is pre-classified.

    \item By construction, $\depth(\preclassified{P}) = \depth(  \grid( \hat{\boundary P} ) ) = \depth(\hat{\boundary P}) = \depth(P)$.

    \item Let $R$ be a pre-classified refinement of $P$, then the set of nodes of $R$ must contain $\hat{\boundary P}$, thus $R$ refines $\preclassified{P}$. 
\end{enumerate}
\end{proof}

\begin{example} \label{example:P3}
We continue Example \ref{example:P1} by constructing the pre-classification $\preclassified{P}$ of the partition $P$ given in Figure \ref{fig:P}. In this case, $a$ is pre-classified, but $b$ is not since $\zeta \in \boundary P$ satisfies $\zeta \subset \xi \in \subcube b$; similarly, $c$ is not pre-classified because $\iota, \kappa \subset \pi$. Therefore, $\preclassified{P}$ is obtained by refining towards 
$\hat{\zeta} = b_3$, $\hat{\iota} = c_5$, and $\hat{\kappa} = c_7$.
The result is displayed in Figure~\ref{figure:preclassifiedP}.
\end{example}

 Our analysis focuses on one partition only and hence is satisfied with the
 partition-centric definition of pre-classified where we impose refinement
 constraints on all adjacent partitions of a partition along the boundary.
 In practice, one will make statements over all partitions of a domain.
 While our formalism makes this constraint a unidirectional one---it holds only
 for the partition of study---the application of the formalism to all partitions
 within a grid thus implies that mesh partition boundaries cut only through
 regular meshes, i.e.~the mesh ``left'' and ``right'' from a cut-through has the
 same result.


\begin{lemma} \label{lem:class_facets}
Let $g$ 
be a pre-classified cell of class $c$ within a
partition $P$. Then $g$ contains at least $c$ facets in $\boundary P$. 
Moreover, if $g$ contains at least $c+1$ facets of
$\boundary P$, then $g$ contains a pair of parallel facets of $\boundary P$.
\end{lemma}

\begin{proof}
Let $g$ be a cell of $P$ of depth $l$ and suppose it contains $u \in \boundary^c P$. Denote $g = [a_1, b_1] \times \dots \times [a_d, b_d]$, then its $(d-1)$-subcubes are of the form $y_i = \{ x \in g : x_i = a_i \}$ or $z_i = \{ x \in g : x_i = b_i \}$ for some $1 \le i \le d$. Without loss, let $u = \{ x \in g : x_1 = a_1, \dots, x_c = a_c \} \in \boundary^c P$. Since $g$ is the unique cell in $P$ containing $u$, then by continuity, $g$ is the unique cell in $P$ containing all the facets $y_i$ for all $1 \le i \le c$ that contain $u$.

Suppose $g$ contains another facet $f \in \boundary P$, then we need to consider three cases. Firstly, if $f = z_i$ for some $1 \le i \le c$, then $y_i$ and $f$ are parallel facets. Secondly, if $f = y_j$ for some $c+1 \le j \le d$, then $g$ contains $f \cap u = \{x \in u : x_j = a_j\} \in \boundary^{c+1} P$, which contradicts the class of $g$. Thirdly, the proof is similar if $f = z_j$ for $c+1 \le j \le d$. 
\end{proof}

\begin{figure}[htb]
 \begin{center}
  \includegraphics[width=0.85\textwidth]{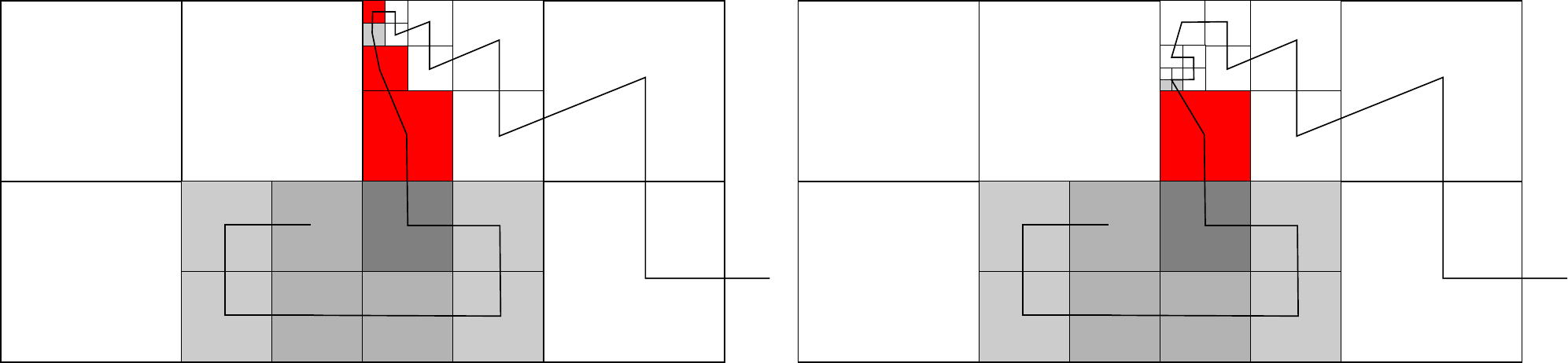}
 \end{center}
 \caption{
  Left:
  A pre-classified partition which is not classified as three cells of class two
  have three domain boundary faces (highlighted). 
  For the grey partition cells, darker means lower cell class.
  Right:
  A second example for pre-classified partitions which are not classified.
  This second example highlights that we do not require 2:1-balancing.
  \label{figure:classified-partitions:not-classified-partitions}
 }
\end{figure}

 \noindent
%
 If a tower of single cells sticks out of a partition and if we refine around
 that ``appendix'' to obtain a pre-classified partition, the cells that stick
 out can host more than $d$ parallel boundary faces and hence navigate us
 into a situation where we require a greater equals (``at least'') rather than
 equals in the Lemma \ref{lem:class_facets} (Figure
 \ref{figure:classified-partitions:not-classified-partitions}).
 The ``stick-out'' effect becomes more significant for higher dimensions but is
 obviously bounded by the mesh depth, i.e.~the number of refinement steps when we construct the
 grid.
 Per resolution level, the number of cells sticking out furthermore is naturally
 bounded through the compactness of the SFC: 
 All SFCs run forth-and-back, i.e.~run in one direction at most $k$ times
 ($k=2$ for Hilbert, $k=3$ for Peano, e.g.), to meet the H\"older continuity in
 the limit.

\subsection{Classified partitions}

A cell $g \in P$ is \Define{classified} if $g$ is pre-classified and contains exactly
$\class(g,P)$ facets. A partition is classified if all its cells are classified. The definition identifies a subset of partition boundary cells
following Lemma \ref{lem:class_facets}. By definition, if $P$ is classified,
then
\[
	\ds(P) = \sum_{g \in P} \class(g, P).
\]

\noindent
 We benefit from the definition that inner cells have class 0.
 It is clear that classified partitions are closed under subdivision.

\begin{lemma} \label{lem:sv_subdivision}
Subdividing a cell of class $c$ of a classified partition $P$ yields a classified partition $P'$ with discrete volume and surface
\begin{align*}
	\dv(P') &= \dv(P) + k^d - 1\\
	\ds(P') &= \ds(P) + c(k^{d-1} - 1).
\end{align*}
\end{lemma}

\begin{proof}
The volume of $P'$ is clear.
For the surface, without loss let $g = [0, k^{-l}] \times \dots \times [0,
k^{-l}]$ be a cell of class $c$ in $P$ (with depth $l$). Without loss 
of generality assume that the element of $\boundary^c_P g$ is $\{  z \in g : z_1 = \dots = z_c = 0 \}$. Then subdividing $g$ yields $k^d$ cells 
\[
    g_x := [x_1 k^{-l-1}, (x_1+1) k^{-l-1} ] \times \dots \times [x_d k^{-l-1}, (x_d + 1) k^{-l-1}],
\]
one for every $x = (x_1, \dots, x_d) \in [k]^d$. Since $\class(g_x) = |\{ i : 1 \le i \le c, x_i = 0 \}|$, we obtain 
\[
	|\{x : \class(g_x) = a\}| = \binom{c}{a} (k-1)^{c-a} k^{d-c}
\]
and the surface of the subdivided cell is
\[
	\sum_{a=0}^d a \binom{c}{a} (k-1)^{c-a} k^{d-c} = c k^{d-1},
\]
thus yielding $\ds(P') = \ds(P) - c + c k^{d-1}$. 
\end{proof}


\begin{rationale}
  We next handle the degenerated cell appendices from Lemma
  \ref{lem:class_facets} within our partitions. 
  Indeed, one single further refinement around a given partition always allows
  us to end up with a partition without the degenerated appendices.  
  Removing degenerated appendices, i.e.~cells with too many partition surface
  faces for their class, allows us to establish sharp estimates.
\end{rationale}

\noindent For any partition $P$, the \Define{classification} of $P$, denoted as
$\classified{P}$, is the refinement of $\preclassified{P}$ obtained by subdividing those cells from
$\preclassified{P}$ that are not classified.

\begin{example} \label{example:P4}
We continue Example \ref{example:P1} by giving the classification of $P$. In $\preclassified{P}$, all cells $b_1, \dots, b_4$ and $c_1, \dots, c_7$ are classified. However, $a$ is not classified, as $\class(a, \preclassified{P}) = 2$ and yet $a$ contains three facets of $\preclassified{P}$, namely $\beta$, $\epsilon$, and $\eta$. Thus, $\classified{P}$ is obtained by simply subdividing $a$. The result is displayed in Figure~\ref{fig:classifiedP}.

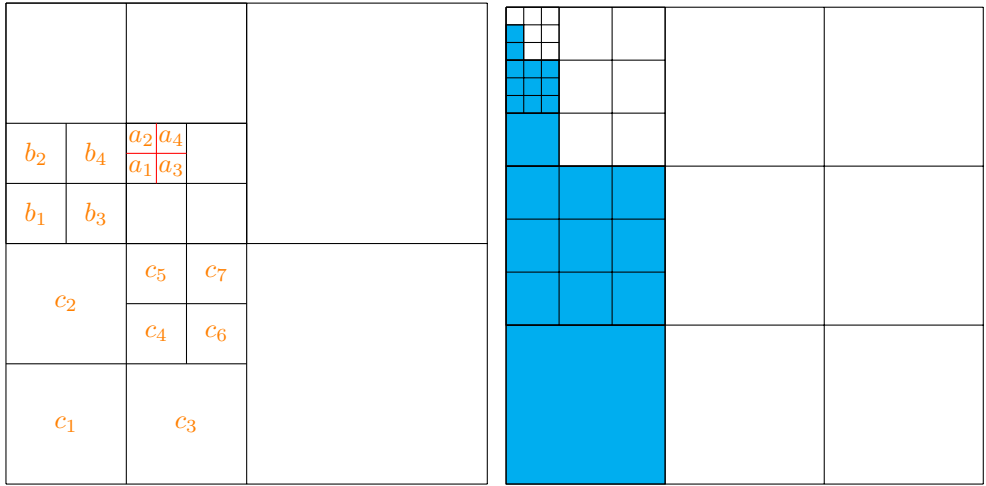
\begin{figure}
\centering
\begin{tikzpicture}[scale=0.8]
\begin{scope}[scale=4]
    \draw (0,0) grid (2,2);
\end{scope}

\begin{scope}[scale=2]
    \draw (0,2) grid (2,4);
\end{scope}

\begin{scope}[scale=1]
    \draw (2,4) grid (4,6);
\end{scope}

\draw (0,2) -- (4,2);
\draw (2,0) -- (2,4);

\draw (0,5) -- (2,5);
\draw (1,4) -- (1,6);

\draw (3,2) -- (3,4);
\draw (2,3) -- (4,3);

\begin{scope}[color=orange]
\draw (2.25, 5.25) node {$a_1$};
\draw (2.25, 5.75) node {$a_2$};
\draw (2.75, 5.25) node {$a_3$};
\draw (2.75, 5.75) node {$a_4$};
\draw (1,1) node {$c_1$};
\draw (1,3) node {$c_2$};
\draw (3,1) node {$c_3$};
\draw (2.5,2.5) node {$c_4$};
\draw (2.5,3.5) node {$c_5$};
\draw (3.5,2.5) node {$c_6$};
\draw (3.5,3.5) node {$c_7$};
\draw (0.5,4.5) node {$b_1$};
\draw (0.5,5.5) node {$b_2$};
\draw (1.5,4.5) node {$b_3$};
\draw (1.5,5.5) node {$b_4$};
\end{scope}

\begin{scope}[color=red]
\draw (2.5,5) -- (2.5,6);
\draw (2,5.5) -- (3,5.5);
\end{scope}
\end{tikzpicture}
\begin{tikzpicture}[scale=0.235]

	\filldraw[fill=cyan, draw=black] (0,0) rectangle (9,9);
	\filldraw[fill=cyan, draw=black] (0,9) rectangle (9,18);
	\filldraw[fill=cyan, draw=black] (0,18) rectangle (3,21);
	\filldraw[fill=cyan, draw=black] (0,21) rectangle (3,24);
	\filldraw[fill=cyan, draw=black] (0,24) rectangle (1,25);
    \filldraw[fill=cyan, draw=black] (0,25) rectangle (1,26);

    \begin{scope}[scale=9]
    	\draw[step = 1cm] (0,0) grid (3,3);
    \end{scope}
    
    \begin{scope}[scale=3]
    	\draw[step = 1cm] (0,3) grid (3,6);
    	\draw[step = 1cm] (0,6) grid (3,9);
    \end{scope}

    \begin{scope}[scale=1]
    	\draw[step = 1cm] (0,21) grid (3,24);
    	\draw[step = 1cm] (0,24) grid (3,27);
    \end{scope}

\end{tikzpicture}
\caption{
 Left: 
 The classification $\classified{P}$ of the partition $P$ on Figure
 \ref{fig:P}. Right:
 Pre-classified partition with linear number of non-classified cells}
 \label{fig:classifiedP}
 \label{fig:linear-non-classified}
\end{figure}

The cells of $\classified{P}$ are listed according to their class as follows:\\
\begin{center}
\begin{tabular}{c|c}
    Class & Cells \\
    \hline
    $0$ & $c_4$\\
    $1$ & $a_1, a_2, b_1, b_3, b_4, c_2, c_5, c_6$ \\
    $2$ & $a_3, a_4, b_2, c_1, c_3, c_7$ \\
\end{tabular}
\end{center}
\end{example}


 The following theorem shows that we can focus on classified
partitions, as long as the cells within a partition do
\emph{not} grow linearly with the mesh depth.
That is, as long as the cell counts grows superlinearly. Its proof is given at the end of this section.

\begin{theorem} \label{th:classification}
For any partition $P$ of depth $M$,
\begin{align*}
	\dv(\classified{P}) &= \dv(P) + O(M),\\
	\ds(\classified{P}) &= \ds(P) + O(M).
\end{align*}
\end{theorem}

 \noindent
 As long as we assume that the finest resolution of the mesh can be found along
 a domain's boundary, we know that the surface $\ds$ is bounded from
 below by $C \sqrt[d]{\dv}$ and from above by $C dk^{-1}$ (see introductory
 example); $C$ is a generic constant.
 The other way
 round, we can say that the volume $\dv$ grows at least linearly in depth. 
 Otherwise we would not hit the upper bound $C dk^{-1}$ from the introduction.
 Without the constraint that the finest mesh has to be found at the boundary, we
 can construct arbitrarily advantageous surface-to-volume ratios.
 As soon as the volume grows faster than linearly,
 the theorem clarifies that we can study the classification of a partition.
 The linear term $O(M)$ quantifies how much we are off from the real data. It 
 eventually disappears with growing mesh depth.

We obtain a further interesting corollary: if the volume
of a partition is superlinear in the depth, then almost all of its cells are classified.

\begin{corollary} \label{cor:non-classified}
Any partition has $O(M)$ non-classified cells.
\end{corollary}

We claim that when discussing the maximum surface-to-volume ratio, one should consider classified partitions only. 
In the example illustrated in Figure \ref{fig:minimal_grid}, the partition $P$ is highlighted in blue, and is given in its minimal grid. Consider the topmost $1$-subcube $y$ of the large blue cell $g$. Even though $y$ is a single element of $\subcube g$, it contains seven elements of $\boundary P$. This contribution then brings the surface-to-volume ratio to $13/2$. This example can be trivially generalised, and we obtain an unbounded surface-to-volume ratio when considering cells that are not classified. This phenomenon occurs because the exterior of the partition is not ``too refined''  for that cell, and as such the cell intersects many small cells outside of the partition. This is a phenomenon similar to that illustrated in Figure \ref{fig:minimal_grid}, and as such we make the choice to disregard it.

%
%
%
%

On the other hand, we have shown in Corollary \ref{cor:non-classified} that there are only $O(M)$ non-classified cells. This is tight, for the example in Figure \ref{fig:minimal_grid} shows a partition with exactly $M$ cells, all of which are not classified; we can even exhibit a pre-classified partition with $O(M)$ cells and $O(M)$ non classified cells in Figure \ref{fig:linear-non-classified}. Thus, when considering volumes that grow superlinearly with the depth ($V(M) \gg M$), the effect of non-classified cells is negligible, and we can just consider classified partitions without loss. 

%
%
%
%
%

 \medskip

We now prove Theorem \ref{th:classification}. The proof consists of four steps. In the sequel, $P$ is a partition of depth $M$ and $Q$ is its shape.
 The first step of the proof is to reduce ourselves to evaluating
 $\dv(\classified{Q})$. 
 Lemma \ref{lem:P*} and its follow-up corollary show that the
 volume of a partition's classified refinement can be characterised by the
 discrete volume of the partition plus the classified refinement along its
 (shape) boundary.
 In a second step, we show that this additional term can be characterised
 through cells that are added due to the pre-classification refinement plus a
 constant term depending on the mesh depth $M$ (Lemma \ref{lem:pre-classified_v_classified}).
 We can ignore cells added by the classification refinement from thereon.
 A third step (Lemma \ref{lem:dv_tilde_Q}) quantifies the remaining term, i.e.~the number
 of cells added at a shape boundary due to the pre-classification.
 The fourth and final step then brings the lemmas together to prove the theorem.

\paragraph{Step 1: Additional cells are found in the classification of the shape}

\begin{lemma} \label{lem:P*}
Let $P$ be a partition and let $\classified{P}$ be its classification. The following hold.
\begin{enumerate}[label=(\alph*)]
	\item \label{it:P**}
	$\classified{P}$ is classified; in fact, $\classified{P}$ is its own classification: $\classified{(\classified{P})} = \classified{P}$.

	\item \label{item:N*}
	$\depth(P) \le \depth(\classified{P}) \le \depth(P) + 1$.
	
	\item \label{it:Pbar}
	A partition is a classified refinement of $P$ if and only if it is a refinement of $\classified{P}$.
	
	\item \label{it:P*_Q*}
	$\classified{P} = \classified{Q} \land P$, where $Q = \shape(P)$.
\end{enumerate}
\end{lemma}

\begin{proof}
~
\begin{enumerate}[label=(\alph*)]
    \item
    Let $g$ be a non-classified cell of $\preclassified{P}$. By Lemma \ref{lem:class_facets}, $g$ contains a pair of parallel facets of $\boundary \preclassified{P}$. After subdividing $g$, any child cell of $g$ is pre-classified and does not contain a pair of parallel facets of $\boundary \classified{P}$, it is hence classified. Moreover, it is easily shown that any classified cell $h$ of $\preclassified{P}$ remains classified in $\classified{P}$. Therefore, $\classified{P}$ is classified.
    As a refinement of a classified does not introduce a
    cell that is not classified, the construction of a
    classified partition is idempotent:
    If $P$ is a classified partition
    already, then the construction yields $P =
    \preclassified{P} = \classified{P}$.

    \item 
    By construction, $\depth(\classified{P}) \le \depth(\preclassified{P}) + 1$, which in turn is equal to $\depth(P) + 1$ by Lemma \ref{lemma:preclassification}.
    
    \item 
    By \ref{it:P**}, $\classified{P}$ is a classified refinement of $P$. Also, any refinement of $\classified{P}$ is also classified by Lemma \ref{lem:sv_subdivision}, and hence it is a classified refinement of $P$. Conversely, if $R$ is a classified refinement of $P$, then $R \refines \preclassified{P}$. Moreover, every non-classified cell of $\preclassified{P}$ must be subdivided in $R$, thus $R$ refines $\classified{P}$.
    
    \item 
    Recall that for all sets of boxes $X$ and $Y$, we denote $X \refines Y$ if $X$ refines $Y$. Then $\classified{P} \refines \classified{Q}$ and $\classified{P} \refines P$, thus $\classified{P} \refines \classified{Q} \land P$. Conversely, $\classified{Q} \land P \refines P$, and since $\classified{Q} \land P$ is classified by Lemma \ref{lem:sv_subdivision}, we have $\classified{Q} \land P \refines \classified{P}$ by \ref{it:Pbar}. Thus, $\classified{P} = \classified{Q} \land P$.
\end{enumerate}
\end{proof}


\begin{corollary} \label{cor:dvP*}
$\dv(\classified{P}) \le \dv(P) + \dv(\classified{Q})$.
\end{corollary}

 \noindent
 Corollary \ref{cor:dvP*} acknowledges that the construction of a
 classification refinement from a partition $P$ affects solely the surface of a
 shape.
 We refine along the shape boundary, but not within the partition.
 Therefore, a classification's discrete volume, i.e.~cell count, is bounded by
 the partition's volume plus additional refinements $\classified{Q}$ around the shape surface.

\begin{figure}
 \begin{center}
  \includegraphics[width=0.5\textwidth]{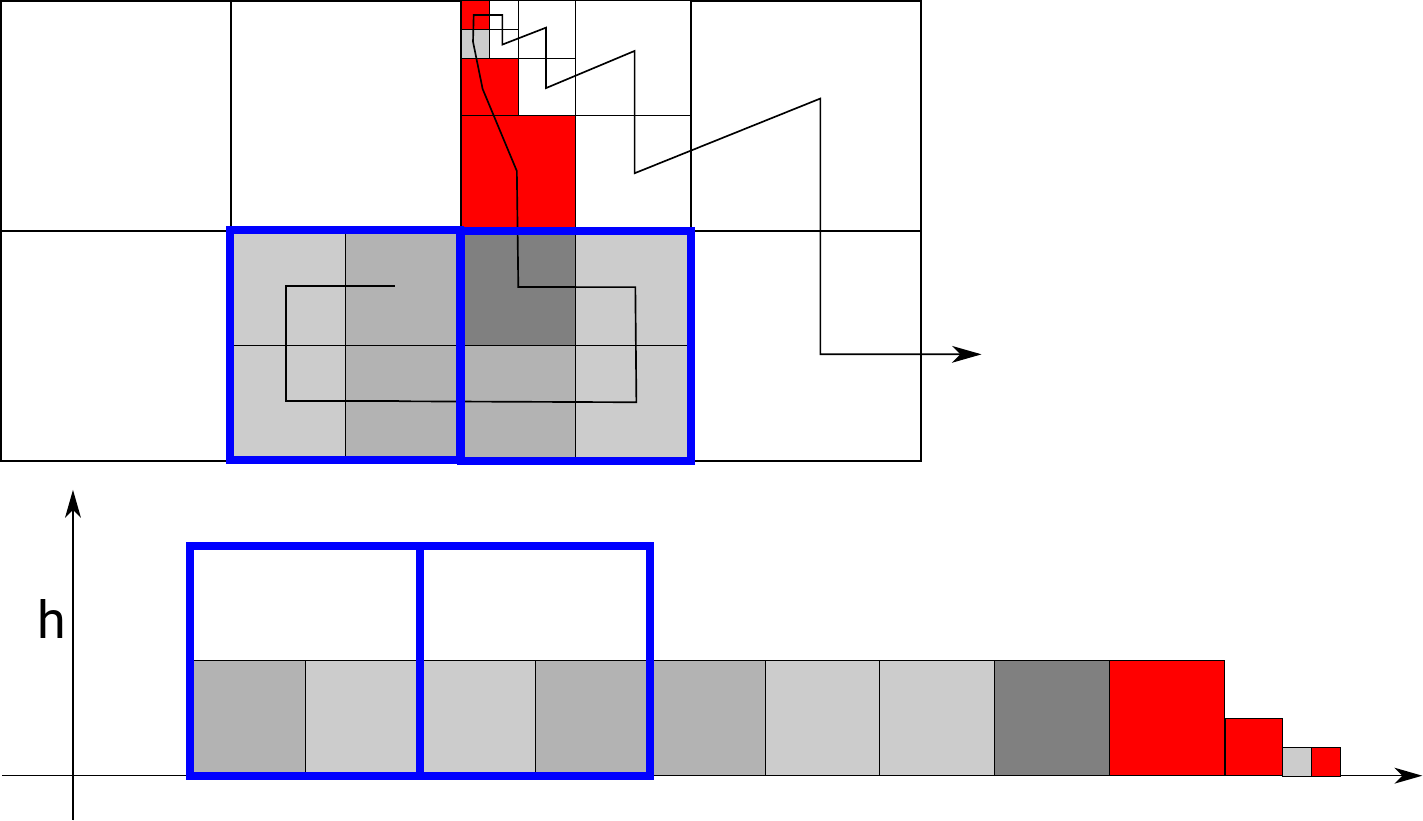}
 \end{center}
 \caption{
  Top: Partition from Figure
  \ref{figure:classified-partitions:not-classified-partitions} with a highlight
  of shape cells (empty blue square).
  Bottom: The shape cells align along the space-filling curve like pearls in a
  row.
  \label{figure:classified-partitions:pearls-in-a-row}
 }
\end{figure}

\paragraph{Step 2: Reducing the problem to evaluating $\dv(\preclassified{Q})$}


By the space-filling property of $\Gamma$, the cells contained in any 
node are consecutive, i.e. they form a partition of $\Gamma$. 
Therefore, $\Gamma$ induces an ordering $Q = (q_1, \dots, q_n)$ of the nodes 
of $Q$. 
Since the cells of a partition are consecutive, consecutive nodes of its shape are adjacent (by the continuity property of $\Gamma$).
 Furthermore, any sequence $Q$ decomposes into two parts whereby one of them
 might degenerate into the empty sequence.
 The sizes of the nodes within the first segment are non-decreasing, while
the sizes within the second segment are non-increasing (Figure
 \ref{figure:classified-partitions:pearls-in-a-row}).

%
%
%
%

\begin{lemma} \label{lem:grid_of_shape}
If $Q$ is the shape of a partition, then $\grid(Q)$ has at most $2 k^d$ cells of depth $l$ for all $l$.
\end{lemma}

\begin{proof}
First, we note that for any $q \in Q$ either $q_1 \subseteq
\parent{q}$ or $q_n \subseteq \parent{q}$. 
Otherwise, we would have $\parent{q} \subseteq \content{Q}$ and hence $q \notin Q$.

Second, we prove that $\grid(Q) = \grid(q_1, q_n)$. Clearly, $\grid(Q) = \bigwedge_{q \in Q} \grid(q)$ refines $\grid(q_1,q_n)$. We now prove the reverse refinement. We have just shown that for any $q \in Q$, $\parent{q}$ is an ancestor of either $q_1$ or $q_n$. But $\parent{q}$ is not a leaf in $\tree(q_1,q_n)$ since the latter contains a descendant of $\parent{q}$ (namely, either $q_1$ or $q_n$), thus $\tree(q_1,q_n)$ also contains the children of $\parent{q}$ and in particular $q$ itself. Thus $\tree(q_1,q_n)$ contains $\tree(Q)$.

We now prove the lemma. For any box $v$ of depth $M$, $\tree(v)$ is a tree with exactly $k^d$ leaves of depth $M$ and $k^d - 1$ leaves of depth $l$ for all $l < M$. Therefore, $\grid(q_1,q_n)$ has at most $2k^d$ cells of any depth.
\end{proof}

\noindent
Our lemma is generous as it argues with the mesh depth $M$.
Indeed, the difference between the resolution level of the coarsest and the
finest node within a partition would yield a stricter bound.

 We use the present lemma in an adaptive mesh context.
 However, it would directly yield the quasi-optimality for regular meshes over
 very fine meshes (cmp.~Appendix \ref{sec:continuous-ratio}).

\begin{lemma} \label{lem:pre-classified_v_classified}
For any partition $P$ with shape $Q$, $\dv(\classified{Q}) =
\dv(\preclassified{Q}) + O(M)$.
\end{lemma}

\begin{proof}
We prove a more general result: $\dv(\classified{P}) = \dv(\preclassified{P}) + O(M)$ for any partition $P$.

We prove that any pre-classified cell $g$ of $\preclassified{P}$ which is not classified actually belongs to $Q = \shape(P)$. Suppose $g$ is pre-classified but not classified, then by Lemma \ref{lem:class_facets} it contains two parallel facets $y$ and $z$ of $\boundary P$. If $g \notin Q$, then  $\parent{g} \subseteq \content{P}$, and in particular, either $y$ or $z$ is in the interior of $\parent{g}$ and hence not in $\boundary P$, which is the desired contradiction. 

Thus, the number of non-classified cells in $\preclassified{P}$ is at most $|Q|$, which is at most $2k^dM$ according to Lemma \ref{lem:grid_of_shape}. Therefore, one needs $O(M)$ refinements to construct $\classified{P}$ from $\preclassified{P}$, and hence $\dv(\classified{P}) = \dv(\preclassified{P}) + O(M)$.
\end{proof}


\paragraph{Step 3: Quantify $\dv(\preclassified{Q})$}

 We continue to give a
 quantitative estimate for $\dv(\preclassified{Q})$. 
 For this, we rely on a simple upper bound on the node counts, i.e.~on fine grid
 cells plus all their coarser predecessors from  the construction process.

\begin{lemma} \label{lem:number_of_nodes}
A grid with $n$ cells has at most $\frac{n}{1 - k^{-d}}$ nodes.
\end{lemma}
\begin{proof}
Let $G$ be a grid and let $V$ be its set of nodes. 
For any $v \in V$, define the height of $v$ as
\[
	\mathrm{height}(v) := \depth(G) - \depth(v).
\]

\noindent
Denote $V_i = \{v \in V: \mathrm{height}(v) = i\}$ for all $i \in \mathbb{N}$ so that $V=V_0 \cup V_1 \cup V_2 \cup \ldots$ with $V_i \cap V_j = \emptyset$
 for $i \ne j$. We prove by induction on $i$ that $|V_i| \le \frac{n}{k^{id}}$. This is clear for $i=0$, since nodes of height $0$ are cells in $G$. Suppose it holds for $i-1$. If $v$ has height $i$, then all its $k^d$ children have height $i-1$, and they are the children of $v$ only, thus
$|V_i| \le \frac{|V_{i-1}|}{k^d} \le \frac{n}{k^{id}}$. We finally obtain
\[
	|V| = \sum_{i=0}^\infty |V_i| \le n \sum_{i=0}^\infty k^{-id} = \frac{n}{1 - k^{-d}}. 
\]
\end{proof}

\noindent

 With this straightforward estimate at hand, we show that the number of cells
 $\nu_g $ in a pre-classified partition can be estimated through the number of
 faces of a partition's boundary.
 The latter means we count the faces of a partition boundary including the
 subdivided fragments.
 In line with Lemma \ref{lem:pre-classified_v_classified}, the resulting Lemma
 \ref{lem:dv_one_cell} and its supporting claim can be phrased over $P$
 which is more general than $Q$.
 However, we apply it for $Q$ only and thus phrase it over $Q$.

\begin{lemma} \label{lem:dv_one_cell}
For any $g \in Q$, 
let $d_g := |\{ z \in \boundary Q: z \subseteq g\}|$ 
and $\nu_g := |\{ r \in \preclassified{Q} : r \subseteq g \}|$.
Then
\[
	\nu_g \le \frac{k}{1 - k^{-d+1}} d_g.
\]
\end{lemma}

\begin{proof}
We first settle the case of one boundary of $g$. For any grid $G$ and any $s \in \subcube^c G$, $G$ naturally induces the grid $G[s] := \{ t \in \subcube^c G : t \subseteq s \}$ on $s$. For any cell $g \in G$ and any $y \in \subcube g$, let $Y = \{ z \in \boundary P, z \subseteq y \}$ and $\hat{Y} = \{ \hat{z} : z \in Y \}$. Let $H_y$ be the minimal grid of $\hat{Y}$ on $g$: $H_y = \grid( \hat{Y} )[g]$. 

\begin{claim} \label{claim:grid_classification}
For any $y \in \subcube g$, $\dv(H_y) \le \frac{k}{1 - k^{-d+1}} \dv(G[y])$.
\end{claim}

\begin{proof}
We place ourselves inside of $g$. Without loss of generality, say $y = \{ x \in g : x_1 = 0 \}$. Say a box $v \subseteq g$ is close to $y$ if $\min\{ x_1 : x_1 \in v \} = \min\{ x_1 : x_1 \in \parent{v} \}$. In $\tree( \hat{Y} )$, any box that is not close to $v$ is a leaf. Then it is easily seen that $\tree(Y)$ is obtained by removing all leaves in $\tree( \hat{Y} )$ corresponding to boxes that are not close to $Y$. Therefore,
\[
    \dv(H_y) \le |\tree( \hat{Y} )| \le k | \tree(Y) | \le \frac{k}{1 - k^{-d+1}} \dv(G[y]).
\]
\end{proof}

\noindent
 We return to the proof of Lemma \ref{lem:dv_one_cell}, where we recognise
 that $\sum_{y \in \subcube g} \dv(G[y]) = d_g$. Let $H = \preclassified{Q}[g]$ with $\nu_g=|H|$. Since $\bigwedge_{y \in \subcube g} H_y \refines H$, we obtain
\[ 
    \nu_g = |H| \le \sum_y |H_y| \le \frac{k}{1 - k^{-d+1}} d_g.
\]
\end{proof}

\begin{lemma} \label{lem:dv_tilde_Q}
$\dv(\preclassified{Q}) = O(M)$. 
\end{lemma}

\begin{proof}
Let $\grid(Q) = G = \{ g_1, \dots, g_n \}$. Firstly, we have
\begin{equation} \label{equation:4dn}
    \sum_{i=1}^n |\{ g_j : g_i \sim g_j \}| = 2 |\{ \{i,j\} : g_i \sim g_j \}| \le 2 |\subcube G| = 4dn.
\end{equation}

For any $1 \le i \le n$, let $d_i := |\{ g_j : g_i \sim g_j \}|$ and $\nu_i := |\{ r \in \preclassified{Q} : r \subseteq g_i \}|$. We have 
\begin{alignat*}{3}
	\dv(\preclassified{Q}) &\le \sum_{i=1}^n \nu_i & \qquad & \textit{Definition of $\preclassified{Q}$}\\
	&\le \frac{k}{1 - k^{-d+1}} \sum_{i=1}^n d_i & \qquad & \textit{Lemma \ref{lem:dv_one_cell}}\\
	&\le \frac{k}{1 - k^{-d+1}} 4dn & \qquad & \textit{Equation \eqref{equation:4dn}}\\
	&\le \frac{k}{1 - k^{-d+1}} 4d 2k^d M & \qquad & \textit{Lemma \ref{lem:grid_of_shape}}.
\end{alignat*}	
\end{proof}

\paragraph{Step 4: Final theorem proof}

We gather all our results and finally prove the theorem.

\begin{proof}[Proof of Theorem \ref{th:classification}]
We prove that $\dv(\classified{P}) = \dv(P) + O(M)$. We have
\begin{alignat*}{3}
	\dv(\classified{P}) &\le \dv(P) + \dv(\classified{Q})   &\qquad& \textit{Corollary \ref{cor:dvP*}}\\
			&= \dv(P) + \dv(\preclassified{Q}) + O(M)   &\qquad& \textit{Lemma  \ref{lem:pre-classified_v_classified}}\\
			&= \dv(P) + O(M)    &\qquad& \textit{Lemma \ref{lem:dv_tilde_Q}}.
\end{alignat*}

Let us now prove the result on the discrete surface. We note that subdividing a cell can only increase the discrete surface by at most $2d k^{d-1}$, thus 
\[
	\ds(\classified{P}) - \ds(P) \le (2d k^{d-1}) (\dv(\classified{P}) - \dv(P)) = O(M).
\]
\end{proof}

\section{Shape-class-regular partitions}
\label{sec:shape-class-regular-partitions}

\subsection{Shape-class-regular partitions for a given shape}

The strategy to obtain the highest surface to volume ratio for a given depth $M$ and volume $V$ is to do as such. First, select a shape $Q$ and classify it. Second, refine towards all the cells of class $d$ (as each contributes $d$ facets towards the surface) up to depth $M$, then refine towards the cells of class $d-1$, and so on until we reach $V$ cells in total. Therefore, in this section we determine the volume and the surface of the partitions obtained in that fashion.

Let $Q$ be a shape of depth $N$, with classification $Q^*$ of depth $N^*$. Then for all $0 \le c \le d$ and all $M \ge N^*$, we introduce the \Define{shape-class-regular partition} $\shapeclassregular_c(Q,M)$ as follows. First, let
\[
    \shapeclasscells(Q,M) := \{ v \subseteq \content{Q} : \depth(v) = M \}
\]
be the boxes of depth $M$ contained in $\content{Q}$. Since $\shapeclasscells(Q,M) \refines Q^*$, it is classified. Then for any $0 \le c \le d$, let 
\begin{align*}
	\shapeclasscells_c(Q,M) &:= \{v \in \shapeclasscells(Q,M) : \class(v, \shapeclasscells(Q,M)) \ge c \}, \\
	\shapeclassregular_c(Q,M) &:= \grid(\shapeclasscells_c(Q,M))[Q],\\
	V_c(Q,M) &:= \dv(\shapeclassregular_c(Q,M)),\\
	S_c(Q,M) &:= \ds( \shapeclassregular_c(Q,M) ),\\
	R_c(Q,M) &:= \dr( \shapeclassregular_c(Q,M) ).
\end{align*}
We then have
\begin{align*}
    \shapeclasscells_0(Q,M) &\supseteq \shapeclasscells_1(Q,M) \supseteq \dots \supseteq \shapeclasscells_d(Q,M),\\ 
	\shapeclassregular_0(Q,M) &\refines \shapeclassregular_1(Q,M) \refines \dots \refines \shapeclassregular_d(Q,M),\\
	V_0(Q,M) &\ge V_1(Q,M) \ge \dots \ge V_d(Q,M),\\
	S_0(Q,M) &\ge S_1(Q,M) \ge \dots \ge S_d(Q,M).
\end{align*}


For $0 \le c \le d-1$, the asymptotic behaviour of $V_c(Q,M)$ and $S_c(Q,M)$ is dictated by $\cv(\boundary^c Q)$. However, this is not the case for $c=d$. Instead, for all $1 \le l \le N^*$ and $0 \le r \le d$, let $A(Q^*, l, r)$ denote the elements in $\boundary^r Q^*$ that belong to a cell of $Q^*$ of depth $l$:
\[
    A(Q^*, l, r) = |\{ s \in \boundary^r Q^* : s \subseteq q^* \in Q^*, \depth(q^*) = l  \}|.
\]
Then let
\[
    \gamma(Q,M) := \sum_{l=1}^{N^*} A(Q^*,l,d) (M- l).
\]

The main result of this section is Theorem \ref{th:VcSc(Q,M)}, where we determine the asymptotic behaviour of $V_c(Q,M)$ and $S_c(Q,M)$; the asymptotic behaviour of $R_c(Q,M)$ then immediately follows. The rest of this section is devoted to its proof.


\begin{theorem} \label{th:VcSc(Q,M)}
Let $Q$ be a shape. For any $0 \le c \le d-1$, we have
\begin{align*}
    V_c(Q,M) &= \frac{k^d - 1}{k^{d-c} - 1} \cv(\boundary^c Q) k^{M(d-c)} + o(k^{M(d-c)}),\\
    S_c(Q,M) &= c \frac{k^{d-1} - 1}{k^{d-c} - 1} \cv(\boundary^c Q) k^{M(d-c)} + o(k^{M(d-c)}).
\end{align*}
For $c=d$, we have
\begin{align*}
    V_d(Q,M) &= (k^d - 1) \gamma(Q,M) + o(M),\\
    S_d(Q,M) &= d(k^{d-1} - 1) \gamma(Q,M) + o(M).
\end{align*}
\end{theorem}

We introduce the notation
\[
    \rho = \rho(k,d) := \frac{k^{d-1} - 1}{k^d - 1}.
\]

\begin{corollary}
Let $Q$ be a shape. For any $0 \le c \le d$, we have
\[
    \lim_{M \to \infty} R_c(Q,M) = c \rho.
\]
\end{corollary}

The main building block for this step is obtained by focusing on one cell only. If the cell in question has class $r$, then we need to consider the \BF{class-regular grid} $\classregular(c,r,M)$ for $0 \le c \le r \le d$, where we repeatedly refine towards all the $c$-subcubes of the cell containing the element of $\boundary^r P$. 

We formally define the class-regular grid as follows; we place ourselves in $\hypercube$. First of all, for any box $v = [v_1 k^{-l}, (v_1 + 1) k^{-l}] \times \dots \times [v_d k^{-l}, (v_d + 1) k^{-l}]$ of depth $l$ and any $1 \le r \le d$, we denote
\[
    \ax(v, r) := |\{i : 1 \le i \le r, v_i = 0\}|.
\]
Then the class-regular grid is 
\[
    \classregular(c,r,M) := \grid \left( \{ v: \depth(v) = M, \ax(v,r) \ge c \} \right).
\]
A few examples of class-regular grids include:
\begin{itemize}
    \item $\classregular(0,r,M)$ is the regular grid of depth $M$;
    
    \item $\classregular(2,2,M)$ is the ``corner grid,'' illustrated in Figure \ref{fig:introduction:corner};
    
    \item $\classregular(1,2,M)$ is the ``two-side grid,'' illustrated in Figure \ref{fig:2sides};
    
    \item $\classregular(1,1,M)$ is the ``side grid,'' illustrated in Figure \ref{fig:2sides}.
\end{itemize}

We now give an explicit characterisation of $\classregular(c,r,M)$. It will be useful to introduce the following set of cells:
\begin{align*}
	\bx(v,r) &:= |\{i : 1 \le i \le r, v_i \le k-1\}|,\\
	\classcells(l,a,b) &:= \{v : \depth(v) = l, \ax(v,r) = a, \bx(v,r) = b \}. 
\end{align*}
Technically, the definition of $\classcells(l,a,b)$ depends on $r$, but the value of the latter will be clear from the context.

\begin{lemma} \label{lem:subdivision}
The cells of the class-regular grid are $\classregular(c,r,M) = \bigcup_{l=1}^M L_l$, where the cells of depth $l$ are given by 
\begin{align*}
	L_l &= \bigcup \{ \classcells(l,a,b) : 0 \le a \le c-1, c \le b \le r\} \quad \text{for } 1 \le l \le M-1,\\
	L_M &=  \bigcup \{ \classcells(M,a,b) : 0 \le a \le r, c \le b \le r \}.
\end{align*}
\end{lemma}

\begin{proof}
A box $v$ contains an element of $\classcells(M,a,b)$ for some $c \le a \le b \le r$ if and only if $\ax(v,r) \ge c$. Thus, starting from the one-cell grid $\{ \mathbb{H} \}$ and subdividing boxes with $\ax(v,r) \ge c$ repeatedly $M$ times yields $\classregular(c,r,M)$.

We prove the result by induction on $M$; we denote the set of cells with depth $l$ in $\classregular(c,r,M)$ as $L_l(M)$. For $M=1$, subdividing the one-cell grid yields the regular grid of depth one. Since any cell $v$ of depth one satisfies $\bx(v,r) = r \ge c$, we obtain
\[
	L_1(1) = \{v : \depth(v) = 1 \} = \bigcup \{ \classcells(1,a,b) : 0 \le a \le r, c \le b \le r \}.
\]
Now suppose it holds for $M$, and subdivide the cells $v$ with $\ax(v,r) \ge c$. These only occur in $L_M(M)$, hence the cells of depth $l$ remain $L_l(M) = L_l(M+1)$ for all $l \le M-1$, and the cells of depth $M$ are exactly $\{v \in L_M(M) : \ax(v,r) \le c-1\} = L_M(M+1)$. Moreover, the subdivided cells produce the following cells of depth $M+1$: $\{v : \depth(v) = M+1, \bx(v,r) \ge c\} = L_{M+1}(M+1)$.
\end{proof}

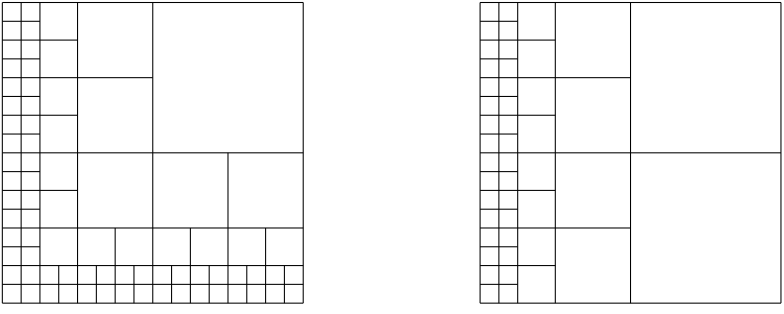
\begin{figure}
\centering
\setlength{\unitlength}{0.5cm}
\begin{picture}(8,8)
\put(0,0){\line(0,1){8}}
\put(0,0){\line(1,0){8}}
\put(8,0){\line(0,1){8}}
\put(0,8){\line(1,0){8}}

\put(4,0){\line(0,1){8}}
\put(0,4){\line(1,0){8}}

\put(2,0){\line(0,1){8}}
\put(0,2){\line(1,0){8}}

\put(6,0){\line(0,1){4}}
\put(0,6){\line(1,0){4}}

\put(1,0){\line(0,1){8}}
\put(0,1){\line(1,0){8}}

\put(7,0){\line(0,1){2}}
\put(0,7){\line(1,0){2}}
\put(3,0){\line(0,1){2}}
\put(5,0){\line(0,1){2}}
\put(0,3){\line(1,0){2}}
\put(0,5){\line(1,0){2}}

\put(0.5,0){\line(0,1){8}}
\put(0,0.5){\line(1,0){8}}

\put(1.5,0){\line(0,1){1}}
\put(2.5,0){\line(0,1){1}}
\put(3.5,0){\line(0,1){1}}
\put(4.5,0){\line(0,1){1}}
\put(5.5,0){\line(0,1){1}}
\put(6.5,0){\line(0,1){1}}
\put(7.5,0){\line(0,1){1}}

\put(0,1.5){\line(1,0){1}}
\put(0,2.5){\line(1,0){1}}
\put(0,3.5){\line(1,0){1}}
\put(0,4.5){\line(1,0){1}}
\put(0,5.5){\line(1,0){1}}
\put(0,6.5){\line(1,0){1}}
\put(0,7.5){\line(1,0){1}}

\end{picture}
\hspace{2cm}
\setlength{\unitlength}{0.5cm}
\begin{picture}(8,8)
\put(0,0){\line(0,1){8}}
\put(0,0){\line(1,0){8}}
\put(8,0){\line(0,1){8}}
\put(0,8){\line(1,0){8}}

\put(4,0){\line(0,1){8}}
\put(0,4){\line(1,0){8}}

\put(2,0){\line(0,1){8}}
\put(0,2){\line(1,0){4}}
\put(0,6){\line(1,0){4}}

\put(1,0){\line(0,1){8}}
\put(0,1){\line(1,0){2}}
\put(0,7){\line(1,0){2}}
\put(0,3){\line(1,0){2}}
\put(0,5){\line(1,0){2}}

\put(0.5,0){\line(0,1){8}}
\put(0,0.5){\line(1,0){1}}
\put(0,1.5){\line(1,0){1}}
\put(0,2.5){\line(1,0){1}}
\put(0,3.5){\line(1,0){1}}
\put(0,4.5){\line(1,0){1}}
\put(0,5.5){\line(1,0){1}}
\put(0,6.5){\line(1,0){1}}
\put(0,7.5){\line(1,0){1}}
\end{picture}
 \caption{
  The two-side adaptive grid $\classregular(c=1, r=2, M=4)$ (left) and the 
  one-side adaptive grid $\classregular(c=1, r=1, M=4)$ (right).
  \label{fig:2sides}
 }
\end{figure}

%
%
%
%

We shall use $\classregular(c,r,M)$ as a building block: we will replace every cell of a given partition by a grid $\classregular(c,r,M)$. As such, only part of the boundary (corresponding to the first coordinates being equal to zero) should be considered in the ``a-surface'' of $\classregular(c,r,M)$. Thus let
\[
	\as(\classregular(c,r,M)) := \sum_{v \in \classregular(c,r,M)} \ax(v,r).
\]

\begin{lemma} \label{lem:class-regular}
For all $0 \le c \le r \le d$ and $M$, the discrete volume and a-surface of $\classregular(c,r,M)$ are given as follows. For $0 \le c \le d-1$, we have
\begin{align*}
	\dv(\classregular(c,r,M)) &=  \frac{k^d-1}{k^{d-c}-1} \binom{r}{c} k^{M(d-c)} + o(k^{M(d-c)}),\\
	\as(\classregular(c,r,M)) &=  c \frac{k^{d-1}-1}{k^{d-c}-1} \binom{r}{c} k^{M(d-c)} + o(k^{M(d-c)}).
\end{align*}
For $c=d$, we have
\begin{align*}
	\dv(\classregular(d,d,M)) &=  (k^d-1)M + o(M),\\
	\as(\classregular(d,d,M)) &=  d (k^{d-1}-1)M + o(M).
\end{align*}
\end{lemma}

\begin{proof}
Let $K = \classregular(c,r,M)$. We have
\begin{align*}
	|\classcells(l,a,b)| &= \binom{r}{a} \binom{r-a}{b-a} (k-1)^{b-a} (k^l - k)^{r-b} k^{l(d-r)}\\
	&= \binom{r}{b} \binom{b}{a} (k-1)^{b-a} (k^l - k)^{r-b} k^{l(d-r)},\\
	|L_l| &= 
	\begin{cases}
		\sum_{a=0}^{c-1} \sum_{b=c}^r |\classcells(l,a,b)| &\text{if } l \le M-1\\
		\sum_{a=0}^r \sum_{b=c}^r |\classcells(M,a,b)| &\text{if } l = M,
	\end{cases}
\end{align*}
hence the discrete volume and a-surface of $K$ are given by
\begin{align*}	
	\dv(K) &= \sum_{l=1}^M \sum_{a=0}^{c-1} \sum_{b=c}^r |\classcells(l,a,b)| + \sum_{b=c}^r \sum_{a=c}^b |\classcells(M,a,b)|,\\
	\as(K) &= \sum_{l=1}^M \sum_{a=0}^{c-1} \sum_{b=c}^r a |\classcells(l,a,b)| + \sum_{b=c}^r \sum_{a=c}^b a |\classcells(M,a,b)|.
\end{align*}

\textbf{Case 1: $0 \le c \le d-1$.} Since $|\classcells(l,a,b)| = \Theta(k^{l(d-b)})$, the discrete volume can be simplified as
\[
	\dv(K) = \sum_{l=1}^M \sum_{a=0}^{c-1} |\classcells(l,a,c)| + |\classcells(M,c,c)| + o(k^{M(d-c)}).
\]
We have
\begin{align*}
	X &:= \sum_{l=1}^M (k^l-k)^{r-c} k^{l(d-r)}\\
	&= \sum_{l=1}^M k^{l(d-r)} \sum_{i=0}^{r-c} \binom{r-c}{i} k^{li} (-k)^{r-c-i}\\
	&= \sum_{i=0}^{r-c} \binom{r-c}{i}  (-k)^{r-c-i}  \sum_{l=1}^M k^{l(d-r)}  k^{li}\\
	&= \sum_{i=0}^{r-c} \binom{r-c}{i}  (-k)^{r-c-i} k^{d-r+i} \frac{k^{M(d-r+i)} - 1}{k^{d-r+i} - 1}\\
	&= \frac{k^{d-c}}{k^{d-c}-1} k^{M(d-c)} + o(k^{M(d-c)}).
\end{align*}
Thus
\begin{align*}
	\dv(K) &= \sum_{a=0}^{c-1} \binom{r}{c} \binom{c}{a} (k-1)^{c-a} X + \binom{r}{c} (k^M - k)^{r-c} k^{M(d-r)} + o(k^{M(d-c)})\\
	&= \binom{r}{c} \frac{k^{d-c}}{k^{d-c}-1} k^{M(d-c)} (k^c - 1) + \binom{r}{c} k^{M(d-c)} + o(k^{M(d-c)})\\
	&= \binom{r}{c} \frac{k^d - 1}{k^{d-c} - 1} k^{M(d-c)} + o(k^{M(d-c)}).
\end{align*}
Similarly, the a-surface is given by
\begin{align*}
	\as(K) &= \sum_{l=1}^M \sum_{a=0}^{c-1} a |\classcells(l,a,c)| + c |\classcells(M,c,c)| + o(k^{M(d-c)})\\
	&= \sum_{a=0}^{c-1} \binom{r}{c} a \binom{c}{a} (k-1)^{c-a} X + c \binom{r}{c} (k^M - k)^{r-c} k^{M(d-r)} + o(k^{M(d-c)})\\
	&= \binom{r}{c} \frac{k^{d-c}}{k^{d-c}-1} k^{M(d-c)} c (k^{c-1} - 1) + c \binom{r}{c} k^{M(d-c)} + o(k^{M(d-c)})\\
	&= c \binom{r}{c} \frac{k^{d-1} - 1}{k^{d-c} - 1} k^{M(d-c)} + o(k^{M(d-c)}).
\end{align*}

\textbf{Case 2: $c = d$.} We then have $K = \classregular(d,d,M)$ and
\begin{align*}
	\dv(K) &= \sum_{l=1}^M \sum_{a=0}^{d-1} \left( \binom{d}{a} (k-1)^{d-a} \right) + 1\\
	& = (k^d-1) M + 1,\\
	\as(K) &= \sum_{l=1}^M \sum_{a=0}^{d-1} \left( a \binom{d}{a} (k-1)^{d-a} \right) + d\\
	&= d(k^{d-1} - 1)M + d.
\end{align*}
\end{proof}

We now determine the volume $V_c(Q,M)$ and the surface $S_c(Q,M)$ for $0 \le c \le d-1$.



\begin{proof}[Proof of Theorem \ref{th:VcSc(Q,M)}]
Let $\classified{Q}$ be the classification of $Q$, with $A(\classified{Q},l,r)$ cells of depth $l$ and class $r$ for all $1 \le l \le \classified{N}$ and $0 \le r \le d$. By Lemma \ref{lemma:cvQ}, we have $\cv(\boundary^c Q) = \cv(\boundary^c \classified{Q})$.

\textbf{Case 1: $0 \le c \le d-1$.} Firstly, a cell of depth $l$ and class $r$ contributes to $r$ facets in $\boundary \classified{Q}$, hence to $\binom{r}{c}$ $(d-c)$-faces in $\boundary^c \classified{Q}$. Its contribution to $\cv( \boundary^c \classified{Q} ) = \cv(\boundary^c Q)$ is then $\binom{r}{c} k^{-l(d-c)}$. Adding up, we obtain
\[
	\cv(\boundary^c Q) = \sum_{l=1}^{\classified{N}} \sum_{r = c}^d A(\classified{Q},l,r) \binom{r}{c} k^{-l(d-c)}.
\]
Secondly, replace each cell of $\classified{Q}$ with depth $l$ and class $r \ge c$ by the class-regular grid $\classregular(c,r,M-l)$ to obtain $\shapeclassregular_c(Q,M)$. The volume and the surface of $\shapeclassregular_c(Q,M)$ are then
\begin{align*}
	\dv(\shapeclassregular_c(Q,M)) &= \sum_{l=1}^{\classified{N}} \sum_{r = c}^d A(\classified{Q},l,r) \dv(\classregular(c,r,M-l)) + O(1)\\
	&= \frac{k^d-1}{k^{d-c} - 1}  \sum_{l=1}^{\classified{N}} \sum_{r = c}^d A(\classified{Q},l,r) \binom{r}{c} k^{(M-l)(d-c)} + o(k^{M(d-c)})\\
	&= \frac{k^d-1}{k^{d-c} - 1} \cv( \boundary^c Q )  k^{M(d-c)} + o(k^{M(d-c)}),\\
	\ds(\shapeclassregular_c(Q,M)) &= \sum_{l=1}^{\classified{N}} \sum_{r = c}^d A(\classified{Q},l,r) \as(\classregular(c,r,M-l)) + O(1)\\
	&= c \frac{k^{d-1}-1}{k^{d-c} - 1}  \sum_{l=1}^{\classified{N}} \sum_{r = c}^d A(\classified{Q},l,r) \binom{r}{c} k^{(M-l)(d-c)} + o(k^{M(d-c)})\\
	&= c \frac{k^{d-1}-1}{k^{d-c} - 1} \cv( \boundary^c Q ) k^{M(d-c)} +  o(k^{M(d-c)}).
\end{align*}
The $O(1)$ term in each summation comes from the cells of class less than $c$ in $\classified{Q}$ which are not subdivided when constructing $\shapeclassregular_c(Q,M)$.

\textbf{Case 2: $c = d$.} By a similar reasoning as above, we have
\begin{align*}
	\dv(\shapeclassregular_d(Q,M)) &= \sum_{l=1}^{\classified{N}} A(\classified{Q},l,d) \dv(\classregular(d,d,M-l)) + O(1)\\
	&= (k^d-1)  \sum_{l=1}^{\classified{N}} A(\classified{Q},l,d) (M-l) + o(M)\\
	&= (k^d-1)  \gamma(Q, M) + o(M),\\
	\ds(\shapeclassregular_d(Q,M)) &= \sum_{l=1}^{\classified{N}} A(\classified{Q},l,d) \as(\classregular(d,d,M-l)) + O(1)\\
	&= d (k^{d-1}-1)  \sum_{l=1}^{\classified{N}} A(\classified{Q},l,d)  (M-l) + o(M)\\
	&= d (k^{d-1} - 1) \gamma(Q, M) + o(M).
\end{align*}
\end{proof}

\subsection{Maximising the discrete volume of shape-class-regular partitions}

The second step of the proof is to maximise the volume of shape-class-regular partitions $\shapeclassregular_c(Q,M)$ over all shapes $Q$ of partitions of the SFC $\Phi$.

We denote the set of all shapes $Q$ of partitions of all DSFCs $\Phi(G)$ as $\Shapes(\Phi)$. We also denote the set of shapes in $\Shapes(\Phi)$ such that $\depth(Q^*) \le M$ as $\Shapes(\Phi,M)$. For $0 \le c \le d$, we define
\[
    V_c(\Phi, M) := \max\{ V_c(Q,M) : Q \in \Shapes(\Phi,M) \}.
\]
We determine the asymptotic behaviour of $V_c(\Phi,M)$. We prove that
\begin{alignat*}{3}
    V_c(\Phi,M) &= \Theta( k^{M(d-c)} ) &\quad& \text{for } 0 \le c \le d-1,\\
    V_d(\Phi,M) &= \Theta(M^2)          &\quad& \text{for } c=d.
\end{alignat*}
In fact, we prove a much tighter result in Theorem \ref{th:Vc(phi)} below. For $0 \le c \le d-1$, we define
\[
    \mu_c(\Phi) := \frac{k^d - 1}{k^{d-c} - 1} \sup\{ \cv(\boundary^c Q) : Q \in \Shapes(\Phi) \}.
\]
(In particular, $\mu_0(\Phi) = (k^d - 1) \cv(\hypercube) = k^d-1$.) For $c = d$, we define
\begin{align*}
    \mu_d^-(\Phi) &:=  (k^d - 1)   \liminf_{M \to \infty} \max\left\{ \frac{\gamma(Q,M)}{M^2} : Q \in \Shapes(\Phi,M) \right\},\\
    \mu_d^+(\Phi) &:=  (k^d - 1)   \limsup_{M \to \infty} \max\left\{ \frac{\gamma(Q,M)}{M^2} : Q \in \Shapes(\Phi,M) \right\}.
\end{align*}

\begin{theorem} \label{th:Vc(phi)}
For $0 \le c \le d-1$, $0 < \mu_c(\Phi) < \infty$ and
\begin{align*}
    V_c(\Phi,M) &= \mu_c(\Phi) k^{M(d-c)} + o(k^{M(d-c)}).
\end{align*}
For $c = d$, $0 < \mu^-_d(\Phi) \le \mu^+_d(\Phi) < \infty$ and
\begin{align*}
    \mu^-_d(\Phi)M^2 + o(M^2) \le  V_d(\Phi,M) &\le \mu^+_d(\Phi) M^2 + o(M^2).
\end{align*}
\end{theorem}

The rest of this subsection is devoted to the proof of Theorem \ref{th:Vc(phi)}. We first prove the bounds on $\mu_c(\Phi)$, $\mu_d^-(\Phi)$, and $\mu_d^+(\Phi)$. We break down the proof into three lemmas, one for each of $\mu_c(\Phi)$, $\mu_d^-(\Phi)$, and $\mu_d^+(\Phi)$.

\begin{lemma} \label{lem:mu1}
For $0 \le c \le d-1$,
\[
    \binom{d}{c} 2^c \le \sup\{ \cv(\boundary^c Q) : Q \in \Shapes(\Phi) \} \le 2k^d \frac{ k^{d-c} } {1 - k^{d-c}}.
\]
Therefore, $0 < \mu_c(\Phi) < \infty$.
\end{lemma}

\begin{proof}
If $Q$ is the shape of a partition and $Q$ has depth $N$, then it contains at most $2k^d$ boxes of depth $1 \le l \le N$ by Lemma \ref{lem:grid_of_shape}. For $c < d$, a cell $q$ of depth $l$ has $\cv(\boundary^c q) = k^{-l(d-c)}$, hence
\begin{align*}
	\cv(\boundary^c Q)
	&\le \sum_{q \in Q} k^{-l(d-c)}\\
	&\le  \sum_{l=1}^N 2k^d k^{-l(d-c)}\\
	&\le 2k^d \frac{ k^{d-c} } {1 - k^{d-c}}.
\end{align*}
Conversely, if $Q = \{\hypercube\}$, then $\cv(\boundary^c \hypercube) = \binom{d}{c} 2^c$ by Equation \eqref{equation:cvg}.
\end{proof}

\begin{lemma} \label{lem:mu2}
For any $M$, there exists $Q \in \Shapes(\Phi,M)$ such that 
\[
    \gamma(Q, M) \ge \frac{1}{4} M^2 + o(M^2). 
\]
Therefore, $\mu_d^-(\Phi) > 0$.
\end{lemma}

\begin{proof}
We define the grid $G_N = \{g_{N,0}, \dots, g_{N,n_N - 1}\}$ (sorted according to $\Phi$) of depth $N$ for all $N \ge 0$ even, as follows. First, $G_0 = \{ g_{0,0} = \hypercube \}$. Then $G_N$ is obtained by replacing $g_{N-2,N/2 - 1}$ by a uniform grid of depth $2$: 
\[
    G_N = \{ g_{N,0} = g_{N-2,0}, \dots, g_{N, N/2 - 2} = g_{N-2, N/2-2}, g_{N, N/2 - 1}, \dots, g_{N, N/2 - 1 + k^{2d} - 1} \},
\]
where $g_{N, N/2 - 1} \cup  \dots \cup g_{N, N/2 - 1 + k^{2d} - 1} = g_{N - 2, N/2 - 1}$. 

We now place ourselves in $G_N$, so we omit the dependence on $N$ and write $G = \{g_0, \dots, g_{n - 1}\}$. Then the shape $Q = Q_N$ is given by $Q= \{g_0, \dots, g_{N/2 - 1}\}$. It is easily verified that:
\begin{enumerate}
    \item $Q$ is indeed a shape;
    \item $Q$ contains the first $N/2$ cells of $G$;
    \item $Q$ is a partition of $\Phi(G)$;
    \item $\depth(g_m) = 2m+2$ for all $0 \le m \le N/2-1$. 
\end{enumerate}
We introduce the sets $D_m := g_0 \cup \dots \cup g_{m-1}$ and $E_m := g_{m+1} \cup \dots \cup g_{N/2-1}$ for any $0 \le m \le N/2-1$.

\begin{claim} \label{claim:no_parallel_faces}
$D_m$ does not contain a pair of parallel facets of $g_m$. 
\end{claim}

\begin{proof}
Suppose that it is indeed the case, i.e. $x$ and $y$ are two parallel facets of $g_m$, where $x = g_m \cap g_i$ and $y = g_m \cap g_j$ for some $i,j < m$. The cell $g_m$ does not contain two parallel facets of its parent $\parent{g_m}$, thus either $g_i$ or $g_j$ (say $g_j$ without loss) contains a child of $\parent{g_m}$. But then $\depth(g_j) \le \depth(g_m) - 2 = \depth(\parent{g_m}) - 1$, and hence $g_m \subseteq \parent{g_m} \subseteq g_j$, which is the desired contradiction.
\end{proof}

\begin{claim} \label{claim:far_away_vertex}
For every $0 \le m \le N/2 - 1$, either $g_m$ or $g_{m+1}$ contains a cell of class $d$ of depth at most $2m+5$ in $\classified{Q_N}$.
\end{claim}

\begin{proof}
We only need to prove that for every $0 \le m \le N/2 - 1$, either $g_m$ or $g_{m+1}$ contains a vertex (i.e. a $0$-face) at distance at least $k^{-2m-5}$ from every other cell in $Q_N$. According to Claim \ref{claim:no_parallel_faces}, $g_m$ shares $r \le d$ facets with $D_m$. 

\textbf{Case 1: $r < d$.} In this case, $g_m$ has at least two vertices $u$ and $v$ that do not belong to $D_m$, and hence they are at distance at least $k^{-2m-2}$ from $D_m$. Also, by construction $E_m \subseteq \parent{g_m}$, a box of size $k^{-2m-3}$. Therefore, either $u$ or $v$ is a vertex of $g_m$ at distance at least $k^{-2m-2}$ from $D_m$ and at distance at least $k^{-2m-3}$ from $E_m$, and we are done. 

\textbf{Case 2: $r = d$.} Since $D_m$ does not contain any pair of parallel facets of $g_m$, it must contain exactly the $d$ facets incident to a given vertex $u$ of $g_m$. Thus the opposite vertex $v$ of $u$ in $g_m$ is at distance at least $k^{-2m-2}$ from $D_m$. If $v \notin E_m$, then $v$ is at distance at least  $k^{-2m-3}$ from $E_m$, and we are done. If $v \in E_m$, then $g_{m+1}$ is at distance at least $k^{-2m-3}$ from $D_m$, and hence $D_{m+1}$ contains only one face of $g_{m+1}$, namely $g_m \cap g_{m+1}$. By Case 1, $g_{m+1}$ thus contains a vertex at least $k^{-2m-5}$ from $D_{m+1}$ and $E_{m+1}$.
\end{proof}

Now let $M \ge 1$, and $N = 2 \lceil M/2 - 1 \rceil$ be the largest even number smaller than $M$. By Claim \ref{claim:far_away_vertex}, either $g_m$ or $g_{m+1}$ contains a cell that contributes to $M-l \ge M- (2m+5)$ to $\gamma(Q_N,M)$. Therefore,
\begin{align*}
	\gamma(Q_N, M) &= \sum_{l=1}^M A(\classified{Q}_N, l,d) (M-l)\\
	&\ge \sum_{m=0}^{N/2-1} (M-(2m+5))\\
	&= \frac{M^2}{4} + o(M^2).
\end{align*}
\end{proof}

\begin{lemma}
For all $Q \in \Shapes(\Phi, M)$, 
\[
    \gamma(Q, M) \le 2^d k^d M^2.
\]
Therefore, $\mu_d^+(\Phi) < \infty$.
\end{lemma}

\begin{proof}
For $c = d$, again each cell in $Q$ of depth $l$ contributes to at most $2^d (M - l)$ in the sum, hence
\begin{align*}
	\gamma(Q, M) &\le \sum_{l=1}^N 2k^d 2^d ( M - l )\\
	&\le 2^d k^d 2 \left(NM - \frac{N(N+1)}{2} \right)\\
	&\le 2^d k^d M^2.
\end{align*}
\end{proof}

We can now prove the theorem.

\begin{proof}[Proof of Theorem \ref{th:Vc(phi)}]
We now need to prove the two displayed equations. First, let $0 \le c \le d-1$. For $M$ large, we have
\begin{align*}
    V_c(\Phi,M) &= \max \{ V_c(Q,M) : Q \in \Shapes(\Phi,M) \}\\
    &= \frac{k^d - 1}{k^{d-c} - 1} k^{M (d-c)} \max \{ \cv( \boundary^c Q ) : Q \in \Shapes(\Phi,M) \} + o( k^{M(d-c)} ),\\
    &= \mu_c(\Phi) k^{M(d-c)} + o( k^{M(d-c)} ),
\end{align*}
where the equations follow Theorem \ref{th:VcSc(Q,M)} and Lemma \ref{lem:mu1} respectively.
Second, let $c = d$. For $M$ large, we similarly have
\begin{align*}
    V_d(\Phi,M) &= \max \{ V_d(Q,M) : Q \in \Shapes(\Phi,M) \}\\
    &= (k^d - 1) \max \{ \gamma(Q,M) : Q \in \Shapes(\Phi,M) \} + o( M ),\\
    &\le \mu^+_d(\Phi) M^2 + o( M^2 ),
\end{align*}
where the equations follow Theorem \ref{th:VcSc(Q,M)} and Lemma \ref{lem:mu2} respectively. The lower bound on $V_d(\Phi, M)$ is proved almost identically.
\end{proof}
\section{Maximum surface-to-volume ratio}
\label{sec:maximum-surface-to-volume-ratio}

We are interested in the maximum surface-to-volume ratio of a partition $P$ for a DSFC $\Gamma$ on a grid $G$. We will fix $\Gamma = \Phi(G)$ for a fixed space-filling curve $\Phi$ and we will maximise for all grids and all partitions. Our main objective is to determine the maximum surface-to-volume ratio $R(V)$ as a function of the discrete volume of the partition. We consider the asymptotic ratio where $V(M)$ is a function of the depth $M$ and $M$ tends to infinity. This allows us to neglect some meaningless residual effects: for the corner adaptive grid, the surface-to-volume ratio tends to $2/3$, which reflects the fact that $\ds \sim 2M$ and $\dv \sim 3M$. Moreover, the asymptotic results should be near the actual values even for small values of $M$, since we usually consider exponential values of $V$: $V(M) \sim C k^{Me}$ for some positive constants $C$ and $e$.

For a given space-filling curve $\Phi$, we denote the set of all partitions $P$ of depth $M$ of all DSFCs $\Phi(G)$ as $\Partitions(\Phi,M)$. For any function $V(M)$, we are interested in the maximum surface-to-volume ratio given that the partition has asymptotically $V(M)$ cells:
\[
	R_\Phi(V) := \sup \left\{ \lim_{M \to \infty} \dr(P) :  P \in \Partitions(\Phi,M), \dv(P) \sim V(M) \right\}.
\]

We write $f(M) \gg g(M)$ if $f$ is asymptotically greater than $g(M)$, i.e. if $g(M) = o(f(M))$. 
In order to simplify the statement of the main theorem, we use the convention $V_{d+1}(\Phi,M) = M$.

\begin{theorem}[Maximum surface-to-volume ratio] \label{th:surface_to_volume_asymptotic}
For any integer $0 \le c \le d$, we have 
\[
	R_\Phi(V) = \begin{cases}
    c \rho  & \text{if } V_{c+1}(\Phi,M) \ll V(M) \le V_c(\Phi,M),\\
	 \left( c - 1 + \frac{1}{\alpha} \right) \rho &\text{if } V(M) = \alpha V_c(\Phi,M), \text{ where } 1 < \alpha < \infty.
	\end{cases}
\]
\end{theorem}

The first main step of the proof is to characterise, for any given shape $Q$, the classified partitions with shape $Q$ and highest surface, for any given possible volume. 






We consider the set $\mathcal{P}(Q,M)$ of classified partitions of shape $Q$ and depth $M$. The largest discrete volume of a partition in $\mathcal{P}(Q,M)$ is obviously given by $V_0(Q,M) = \dv( \shapeclassregular_0(Q,M) )$. On the other hand, the smallest discrete volume in $\mathcal{P}(Q,M)$ is denoted as $V_{d+1}(Q,M) = \dv(\classified{Q}) + (k^d - 1)(M - N^*)$, where $N^* = \depth(\classified{Q})$.  Then the volume of a partition in $\mathcal{P}(Q,M)$ can take any value in the set
\[
    \mathcal{V}(Q,M) := \{ V : V_{d+1}(Q,M) \le V \le V_0(Q,M), V  \equiv V_{d+1}(Q,M) \mod k^d - 1 \}.
\]
For any $V \in \mathcal{V}(Q,M)$ with $V_{c+1}(Q,M) \le V \le V_c(Q,M)$, we let $H(Q,V,M) \in \mathcal{P}(Q,M)$ be a partition such that $\dv( H(Q,V,M) ) = V$ and
\[
    \shapeclassregular_c(Q,M) \preceq H(Q,V,M) \preceq \shapeclassregular_{c+1}(Q,M).
\]
We note that $H(Q,V,M)$ is obtained from $\shapeclassregular_{c+1}(Q,M)$ by repeatedly subdividing cells of class $c$. 

\begin{lemma} \label{lem:H(Q,V,M)}
For any $V \in \mathcal{V}(Q,M)$ and any classified partition $R$ of shape $Q$ and discrete volume $\dv(R) = V$, we have $\ds(R) \le \ds(H(Q,V,M))$ and hence $\dr(R) \le \dr( H(Q,V,M) )$.
\end{lemma}

\begin{proof}
Let $h_1, \dots, h_s$ be the respective classes of the boxes that have to be subdivided to obtain $H = H(Q,V,M)$ from $\classified{Q}$. Since $\dv(R) = \dv(H)$, $R$ is also obtained from $\classified{Q}$ after $s$ subdivisions, say of boxes of classes $r_1, \dots, r_s$. By Lemma \ref{lem:sv_subdivision}, we have 
\begin{align*}
	\dv(H) &= \dv(\classified{Q}) + (k^d - 1) \sum_{i=1}^s h_i,\\
	\dv(R) &= \dv(\classified{Q}) + (k^d - 1) \sum_{i=1}^s r_i.
\end{align*} 
By construction, we have $h_s = \min\{ h_i : 1 \le i \le s \}$. Thus, if $\ds(R) > \ds(H)$, there must be a box $b$ of class $r_j > h_s$ which is not subdivided in $H$. This implies in particular that $h_s < d$, and that $V > V_{h_s + 1}(Q,M) \ge V_{r_j}(Q,M)$. But then $b$ contains an element of $\shapeclasscells_{r_j}(Q,M)$, and by construction $b$ must be subdivided in $H$, which is the desired contradiction.
\end{proof}

The second step of the proof is to determine the asymptotic surface-to-volume ratio of $H(Q,V,M)$ in general.

\begin{lemma} \label{lemma:RQV}
Let $R_Q(V) = \lim_{M \to \infty} \dr( H(Q,V,M)  )$. Then for $0 \le c \le d$,
\[
	R_Q(V) = \begin{cases}
    c \rho  & \text{if } V_{c+1}(Q,M) \ll V(M) \le V_c(Q,M),\\
	 \left( c - 1 + \frac{1}{\alpha} \right) \rho &\text{if } V(M) = \alpha V_c(Q,M), \text{ where } 1 < \alpha < \infty.
	\end{cases}
\]
\end{lemma}

\begin{proof}
We have already proved the result for $V(M) = V_c(Q,M)$. We now prove for $V_{c+1}(Q,M) \ll V(M) \le V_c(Q,M)$. Let $H = H(Q,V,M)$, then $H$ is obtained from $\shapeclassregular_{c+1}(Q,M)$ by subdividing cells of class $c$. Denoting $V'(M) = V(M) - V_{c+1}(Q,M)$, there are 
\[
    s = \frac{1 }{ k^d - 1 } V'(M) 
\]
such subdivisions. We obtain 
\begin{align*}
    \dv(H) &= V'(M) + o(V'(M))\\
    \ds(H) &= S_{c+1}(Q,M) + s c (k^{d-1} - 1) =  c \rho V'(M) + o(V(M)),\\
    \dr(H) &= c \rho + o(1).
\end{align*}
We now prove the result for $V(M) = \alpha V_c(Q,M)$ for $\alpha > 1$. Let $H = H(Q,V,M)$, then $H$ is obtained from $\shapeclassregular_c(Q,M)$ by subdividing cells of class $c-1$; there are 
\[
    s = \frac{ V(M) - V_c(Q,M) }{ k^d - 1 } = \frac{ \alpha - 1 }{ k^d - 1 } V_c(Q,M)
\]
such subdivisions. We obtain 
\begin{align*}
    \dv(H) &= \alpha V_c(Q,M)\\
    \ds(H) &= S_c(Q,M) + s c (k^{d-1} - 1) = c \rho V_c(Q,M) + (\alpha - 1) \rho c V_c(Q,M) + o(V_c(Q,M)),\\
    \dr(H) &= \left( c - 1 + \frac{1}{\alpha} \right) \rho + o(1).
\end{align*}

\end{proof}

\begin{proof}[Proof of Theorem \ref{th:surface_to_volume_asymptotic}]
We first prove that the quantity in the right hand side is an upper bound on the surface-to-volume ratio. Consider a sequence $P_M$ of partitions of discrete volume $V(M)$ shape $Q_M$ and depth $M$. Then $\dr(P_M) \le \dr(H_M)$, where $H_M = H( Q_M, V(M), M )$. If $V_{c+1}(\Phi, M) \ll V(M)$, then $V_{c+1}(Q,M) \ll V(M)$, hence $\dr(P_M) \le \dr(H_M) \le c \rho + o(1)$. If $V(M) = \alpha V_c(\Phi,M)$ with $\alpha \ge 1$, then $V(M) \ge \alpha V_c( Q,M )$ and hence $\dr(P_M) \le \dr(H_M) \le \left( c - 1 + \frac{1}{\alpha} \right) \rho + o(1)$.

We now prove that the upper bound can be reached. First, the case $c = 0$ is trivial, as the upper bound is zero. Second, let $1 \le c \le d-1$. For any $\epsilon > 0$, let $Q$ be a shape such that $V_c(Q,M) > (1 - \epsilon) V_c(\Phi,M) + o(1)$. For $V_{c+1}(\Phi,M) \ll  V(M) \le V_c(\Phi,M)$, the sequence $H_M = H(Q,V(M),M)$ satisfies (with $V'(M) = V(M) - V_{c+1}(\Phi,M)$)
\begin{align*}
    \dv( H_M ) &= V'(M) + V_{c+1}(\Phi,M) = (1 + o(1))V'(M)\\
    \ds( H_M ) &= S_{c+1}(Q,M) + c \rho V'(M) = (c \rho + o(1)) V'(M),\\
    \dr( H_M ) &= c \rho + o(1).
\end{align*}
For $V(M) = \alpha V_c(\Phi,M)$ and $\alpha > 1$, the sequence $H_M = H(Q,V(M),M)$ then satisfies 
\begin{align*}
    \dv( H_M ) &= \alpha V_c(\Phi,M) < \frac{ \alpha }{1 - \epsilon} V_c(Q,M),\\
    \ds( H_M ) &\ge \ds( H(Q, \alpha V_c(Q,M), M  ) ) = \left( c - 1 + \frac{1}{\alpha} \right) \rho \alpha V_c(Q,M),\\
    \dr( H_M ) &> (1-\epsilon) \left( c - 1 + \frac{1}{\alpha} \right) \rho + o(1).
\end{align*}

Third, for $c = d$, the proof is similar, but uses a shape $Q$ that maximises $\gamma(Q,M)$ instead. As such, we omit it.
\end{proof}


We would like to make further remarks on the maximum surface-to-volume ratio that complement the formula in Theorem \ref{th:surface_to_volume_asymptotic}.

Firstly, for $V = O(M)$, the maximum surface-to-volume ratio can be higher than $d \cdot \rho$. 
Indeed, the classified partition of depth $M$ ($M = 5$ in Figure \ref{fig:minimal_grid}, right) 
has discrete volume $4M-4$ and discrete surface $5M-2$. Therefore, the surface-to-volume ratio tends to $5/4$, while $d \cdot \rho = 2/3$.

%
%
%
%
%
%

Secondly, the case where $V(M) = \Theta(k^{Md})$ is linked to the problem of maximising the continuous surface-to-volume ratio of shapes. Let us illustrate this link by first considering the shape-regular partition $\shapeclassregular_0(Q,M)$, which contains all cells of depth $M$ in $\content{Q}$. For the sake of simplicity, we shall omit $Q$ and $M$ from our notation in this paragraph. The discrete volume of $\shapeclassregular_0$ is given by $V_0 = \cv( Q ) k^{Md}$. However, the partition $\shapeclassregular_1$ already contains all cells of depth $M$ that contribute to the discrete surface; thus any partition $H$ refining $\shapeclassregular_1$ has the same discrete surface as $\shapeclassregular_1$. In particular, we obtain 
\[
    S_0 = S_1 = \cs(Q) k^{M(d-1)} + o( k^{M(d-1)} ).
\]
Let us denote the continuous surface-to-volume ratio of $Q$ as $r = \frac{ \cs(Q) }{ \cv(Q) }$. We obtain
\[
    \dr(\shapeclassregular_0) = \frac{ \cs(Q) +o(1) }{ \cv(Q) } k^{-M} \sim r k^{-M}.
\]
More generally, for $V = \frac{1}{\beta} V_0$ with $\beta \ge 1$, any $H = H(Q,V,M)$ refines $\shapeclassregular_1$. We thus obtain
\[
    \dr(H) = \beta \dr(\shapeclassregular_0) = \beta r k^{-M}.
\]
This illustrates that the maximum surface-to-volume ratio actually decreases exponentially with the depth $M$, and that the influence of the space-filling curve lies in the continuous surface-to-volume ratio of its partitions.
\section{Conclusion, discussion and outlook}
\label{sec:conclusion}

Our manuscript offers a qualitative and quantitative description of the
worst-case surface-to-volume ratio of partitions of adaptive Cartesian meshes as
they are induced by space-filling curves.

 If we plot this ratio, the plot
starts from a plateau, descends in
smoothed-out steps, before it transitions into a fade-out regime that we know from proofs on regular Cartesian
grids.
The thresholds, i.e.~locations of the staircase steps, are formalised via curve-
and dimension-dependent formulae.
It is remarkable that, once
the dimension of the domain and the way the cells are refined are fixed, the
actual space-filling curve used has a relatively little impact on the maximum surface-to-volume ratio, i.e.~the shape of the curve.

%
%
The fade-out of the curves for (very) large $V(M)$ or mesh depths
$M$, respectively summarises well-known insight on the impact of the H\"older
continuity of SFCs onto the discrete SFC partitions for regular grids.
Domain decomposition for regular grids is not a particularly hard computational
challenge.
The ``interesting'' insight thus is formalised by the curve's staircase pattern.

%
%
Though SFCs are popular tools in scientific computing, they have drawbacks.
First, space-filling curves are no 1:1 fit to load
re-balancing of processes that are dominated by diffusion or waves:
The elegance of SFCs is that they reduce the domain decomposition problem into a
one-dimensional challenge.
This advantage in turn means that we have no control over the shape in the
$d$-dimensional space:
If a wave makes an adaptivity pattern travel through the domain,
e.g., and rebalancing becomes necessary due to this change in the adaptivity, an
SFC-based rebalancing triggers global rebalancing among all SFC partitions.
SFC partitions can only grow or shrink along the curve, but waves do not travel along a space-filling curve.
Consequently, any rebalancing is not localised but all curve partitions in the domain have to change 
even though adding or removing few cells along a wave propagation direction to
a few partitions would already yield a proper repartitioning.
Second, SFC-based splitting runs risk with a high probability that
we cut through the mesh in very adaptive mesh regions.
If particles cluster in a certain subarea of the domain, e.g., and are therefore
resolved with a fine mesh, the SFC cuts very likely run through mesh faces of
the finest resolution.
The aggressive worst-case refinement along domain boundaries from our proof are
thus far from academic.
Finally, many multiscale and tree codes have to work with complete local trees
per compute unit, i.e.~they take a local partition and then enrich this
partition---subject to cell labels for local and augmented cells---until they
get a full (local) spacetree
\cite{Bader:2013:SFCs,Bungartz:2006:ParallelAdaptivePDESolver,Clevenger:2020:MultigridFEM,Suh:2020:Balancing,Sundar:08:BalancedOctrees,Weinzierl:2019:Peano}.
Tail-like, local cuts similar to the tower of cells (Figure \ref{fig:minimal_grid}) thus
not only yield a disadvantageous surface-to-volume ratio, they also induce large
overhead of additional cells.

%
%

In an era of numerical simulations where meshes change quickly
and the actual compute speed is co-determined by power considerations, bandwidth congestion, cache thrashing, and so forth,
i.e.~undeterministic to some degree, a natural choice is to apply SFC cuts on a
rather coarse resolution level of the mesh.

 In our own language, we partition the mesh in rather coarse shapes and further
 refinement levels then ``inherit'' this partitioning.
 The partitioning is coarse grain, and we
accept some ill-balancing.
This mitigates the three shortcomings from above. 
%
%
 This frequently used ``fix'' or realisation pattern however 
requires the developer to
make two choices: on which resolution level to apply the SFC cuts, i.e.~what is the
finest shape, and how often do we want to rebalance?
 If meshes change, existing mesh partitions always can ill-balanced.
 With a coarse grained split-up, these ill-balances quickly become more severe
 or are inherent; we do not give a partitioning the mesh granularity to
 fine-balance work.
 More frequent re-balancing thus becomes mandatory.
While both decisions---mesh grainess and rebalancing frequency---are primarily guided by work balancing considerations, it makes sense to add a proper communication penalty to the cost metrics.
Our analysis yields both qualitative and quantitative guidelines for them:
(i) The classification of a cell should feed into the cost metric.
(ii) If an adaptivity pattern introduces new cells with high class---it ``hits''
a vertex, e.g.---it is reasonable to rebalance. If it introduces additional
cells with low class, it is less urgent to rebalance.
(iii) If the volume of a partition underruns the step thresholds, it it
reasonable to merge this partition with further partitions.

%
%

While this list likely is not comprehensive, it shows the mindset how an
analytic model can feed into load balancing decisions.
The exact study of the interplay of rebalancing algorithms is subject of future research.

\section*{Acknowledgements}

Tobias' work is sponsed by EPSRC under the Excalibur Phase I call. The grant number is EP/V00154X/1 (ExaClaw).
A fundamental goal in this project is the effective
parallelisation of a dynamically adaptive mesh refinement
solver for wave questions which uses the Peano space-filling curve.
Tobias also receives funding from the same programme under grant EP/V001523/1
(Massively Parallel Particle Hydrodynamics for Engineering and Astrophysics)
 where we study Lagrangian formalisms which are embedded into (space-)trees to
 efficiently truncate and approximate short- and long-range interactions between
 particles.

\appendix
\section{Continuous surface-to-volume ratio}
\label{sec:continuous-ratio}

Optimality statements for space-filling curves are well-known for regular grids
with a very high level of detail
\cite{Bungartz:2006:ParallelAdaptivePDESolver,Zumbusch:2001:QualityOfSFCs}.
We use the appendix to show that the known optimality bounds are a special case 
directly resulting from our Lemma \ref{lem:grid_of_shape}.
For this we use the $l^\infty$-norm to define the distance between two
points in the hypercube: for any $x, y \in \hypercube$,
\[
    d_\infty(x,y) :=  \max\{ |x_i - y_i| : 1 \le i \le d \} .
\]
This choice of norm is immaterial, as all norms on $\hypercube$ are equivalent. For any finite set of boxes $X$, the diameter of $X$ is the maximum distance between any two points in $X$:
\[
    \diameter(X) := \max \{ d_\infty(x,y) : x,y \in \content{X} \}.
\]

For a single cell $g$ of depth $M$, we have $\diameter(g) = k^{-M}$ and hence
$\cv( \boundary^c g ) = \binom{d}{c} 2^c k^{-M(d-c)} = \Theta( \diameter(g)^{d-c} )$ for all $0 \le c \le d-1$. Our main theorem is that 
\[
    \cv( \boundary^c P ) = \Theta( \diameter(P)^{d-c} )
\]
holds for any partition $P$.

\begin{theorem} \label{theorem:continuous}
Let $\Phi$ be a SFC and let $P \in \Partitions(\Phi)$. Then for every $0 \le c \le d-1$,
\[
    \frac{1}{U} \diameter(P)^{d-c} \le   \cv( \boundary^c P ) \le U \diameter(P)^{d-c},
\]
where $U = 2 k^d (1 - k^{c-d})^{-1}$.
\end{theorem}

\begin{proof}
Let $Q = \shape(P)$. Since $\diameter(P) = \diameter(Q)$ and $\cv(\boundary^c P) = \cv(\boundary^c Q)$, we only need to prove the result for $Q$. For all $l$, let $A_l$ denote the number of nodes of depth $l$ in $Q$. By Lemma \ref{lem:grid_of_shape}, $A_l \le 2 k^d$ for all $l$. Let $L = \{l : A_l \ge 1 \}$ be the set of possible depths of nodes in $Q$ and let $l_{\min}$ be the minimum element of $L$.

Firstly, let $q \in Q$ have depth $l_{\min}$, then for any two vertices $x$, $y$ of $q$ we have $\diameter(Q) \ge d_\infty(x,y) = k^{-l_{\min}}$. We obtain
\begin{align*}
    \cv( \boundary^c Q) &= \sum_{l \in L} A_l k^{-l(d-c)}\\
    &<  2 k^d k^{-l_{\min} (d-c)} \sum_{i = 0}^\infty k^{-i(d-c)}\\
    &= U k^{-l_{\min} (d-c)}\\
    &\le U \diameter(Q)^{d-c}.
\end{align*}

Secondly, let $x, y \in \content{Q}$ maximise $d_\infty(x,y)$, then $\diameter(Q) = d_\infty(x,y) \le \cv( \boundary^{d-1} Q ) \le U k^{-l_{\min}}$. We obtain
\begin{align*}
    \cv( \boundary^c Q) &= \sum_{l \in L} A_l k^{-l(d-c)}\\
    &\ge k^{-l_{\min} (d-c)}\\
    &\ge \frac{1}{U} \diameter(Q)^{d-c}.
\end{align*}

\end{proof}

In particular, we obtain the upper bound on the continuous surface against the continuous volume in \cite{Zumbusch:2001:QualityOfSFCs}.

\begin{corollary}
Let $P$ be a partition, then 
\[
    \cs(P) \le C \cv(P)^{1 -\frac{1}{d}}
\]
for some constant $C > 0$.
\end{corollary}

In fact, we can even prove H\"older continuity thanks to Lemma \ref{lem:grid_of_shape}. Recall that a function $\Psi : [0,1] \to \hypercube$ is \Define{H\"older continuous} if there exists $U > 0$ such that for all $x,y \in [0,1]$,
\[
   d_\infty(\Psi(x) - \Psi(y)) \le U |x-y|^{1/d}.
\]


\begin{theorem}
Any continuous SFC is H\"older continuous. 
\end{theorem}

\begin{proof}
Let $\Psi$ be a continuous SFC. Let $M$ such that $\Psi(x)$ and $\Psi(y)$ belong to different cells of the regular grid of depth $M$, say $x \in g_i$ and $y \in g_j$ for $i < j$. We then have $|x - y| \ge k^{-Md} (j - i - 1)$.

Let $P = \{g_i, \dots, g_j\}$ be the smallest partition containing $\Psi(x)$ and $\Psi(y)$, and let $Q = \shape(P)$ with $A_l$ nodes of depth $l$, and again $l_{\min}$ being the smallest $l$ such that $A_l > 0$. Since each node in $Q$ of depth $l$ contains $k^{d(M-l)}$ cells in $P$, we obtain
\[
    j - i + 1 = \sum_l A_l k^{d(M-l)},
\]
and hence $|x - y| \ge k^{-dl_{\min}} - 2k^{-Md}$. Thanks to the proof of Theorem \ref{theorem:continuous}, we obtain
\begin{align*}
    d_\infty( \Psi(x), \Psi(y) )^d &\le \cv( \boundary^{d-1} Q )^d\\
    &\le U^d k^{-d l_{\min}}\\
    &\le U^d (|x - y| + 2k^{-Md}).
\end{align*}
Since this is true for all $M$ large enough, we obtain the result.
\end{proof}

\section{Examples of surface-to-volume ratios}

 A lot of grids that we can generate through a recursive construction rule allow
 us to write down the surface-to-volume explicitly and, hence, to validate our
 statements.
 Examples for this are the $K(0,0,M)$ grid which denotes a regular mesh,
 the $K(1,1,M)$ which is a mesh where we refine aggressively towards one face
 (Figure \ref{fig:2sides}) or its natural extension $K(1,2,M)$ where we refine
 towards two faces (Figure \ref{fig:2sides} as well).
 $K(2,2,M)$ finally is the corner-refined grid from Figure
 \ref{fig:introduction:corner}.

The Cantor grid $\Cantor(M)$ finally is another interesting case,
as it does not uniformly refine towards a subcube.
It is obtained by repeatedly subdividing towards the Cantor set of the $\{x_1 =
0\}$ face of the square. More precisely, let $k=3$, $d=2$ and $C$ be the Cantor set, then denote $C' = \{ x \in \hypercube : x_1 = 0, x_2 \in C \}$. The Cantor grid $\Cantor(M)$ is then defined as $\Cantor(0) = \{ \hypercube \}$, and then $\Cantor(M)$ is obtained from $\Cantor(M-1)$ by refining cells that contain an element of $C'$.

The results are displayed in Table \ref{table:ds_dv}. We remark that for the Cantor grid, $\ds, \dv = \Theta( 3^{(\log_3 2) M} )$, which is consistent with the Hausdorff dimension of the Cantor set.

%
%

\begin{figure}[hb]
\centering
\begin{tikzpicture}[scale=0.05]
\begin{scope}[scale=27]
\draw (0,0) grid (3,3);
\end{scope}

\begin{scope}[scale=9]
\draw (0,0) grid (3,3);

\draw (0, 6) grid (3,9);
\end{scope}


\begin{scope}[scale=3]
\draw (0,0) grid (3,3);
\draw (0, 6) grid (3,9);

\draw (0,18) grid (3,21);
\draw (0, 24) grid (3,27);
\end{scope}

\end{tikzpicture}

\caption{
 The Cantor grid for $d=2, k=3$ 
 \label{fig:cantor}
}
\end{figure}
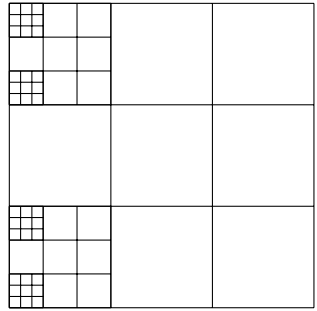

\begin{table}[htb]
\centering
\begin{tabular}{|c|c|c|c|c|c|c|}
    \hline
    Grid & $k$ & $d$ & $\rho$ & $\dv$ & $\ds$ & $\lim_{M \to \infty} \dr$ \\
    \hline
    $\classregular(0,0,M)$  & $2$ & $2$ & $1/3$ & $4^M$ & $2 \cdot 2^M$ & $0 = 0 \cdot \rho$\\
    \hline
    $\classregular(1,1,M)$  & $2$ & $2$ & $1/3$ & $3 \cdot 2^M - 2$ & $2^M + 2M + 4$ & $1/3 = 1 \cdot \rho$\\
    \hline
    $\classregular(1,2,M)$  & $2$ & $2$ & $1/3$ & $6 \cdot 2^M - 3M - 5$ & $2 \cdot 2^M + 2M + 2$ & $1/3 = 1 \cdot \rho$\\
    \hline
    $\classregular(2,2,M)$  & $2$ & $2$ & $1/3$ & $3M + 1$ & $2M + 6$ & $2/3 = 2 \cdot \rho$\\
    \hline
    $\Cantor(M)$                & $3$ & $2$ & $1/4$ & $8 \cdot 2^M - 7$ & $2 \cdot 2^M + 4M + 4$ & $1/4 = 1 \cdot \rho$\\
    \hline
\end{tabular}
\caption{Discrete surface and volume of some grids} \label{table:ds_dv}
\end{table}

\ifthenelse{\boolean{arxiv}}{
 \bibliographystyle{plain}
}{
 \bibliographystyle{siamplain}
}
\bibliography{./paper}

\end{document}